\documentclass[11pt]{amsart}  
\usepackage{amssymb,amsfonts,graphicx,amsmath,amsthm,amstext,latexsym,amsfonts}
\usepackage{graphicx,epsfig}
\usepackage{bbm}
\usepackage{t1enc}
\usepackage{mathrsfs}
\usepackage{subfig}
\usepackage{hyperref}
\usepackage{a4wide}

\usepackage{graphicx}        
\usepackage{tikz}
 \usetikzlibrary{arrows}
 \usetikzlibrary{calc}
\usepackage{pgfplots}
\usetikzlibrary{intersections}
\theoremstyle{plain}

\numberwithin{equation}{section}
\newtheorem{theorem}{Theorem}[section]
\newtheorem{proposition}[theorem]{Proposition}
\newtheorem{lemma}[theorem]{Lemma}
\newtheorem{remark}[theorem]{Remark}

\usepackage{color}
\definecolor{darkred}{rgb}{0.8,0,0}
\definecolor{darkblue}{rgb}{0,0,0.7}
\definecolor{darkgreen}{rgb}{0,0.4,0}

\usepackage{cite}
\newcommand{\R}{{\mathbb R}}

\newcommand{\W}{{\mathcal W}}
\newcommand{\MMM}{\color{black}}
\newcommand{\KKK}{\color{black}}

\newcommand{\tr}{{\rm Tr}}
\newcommand{\xxi}{{\mbox{\boldmath$\xi$}}}

\def\u{\mathbf{u}}
\def\uu{\mathbf{uu}}
\def\vv{\mathbf{v}}
\def\yy{\mathbf{y}}
\def\n{\mathbf{n}}

\def\xx{\mathbf{x}}

\def\xxi{\boldsymbol{\xi}}
\def\uu{\mathbf{u}}
\def\vv{\mathbf{v}}
\def\ww{\mathbf{w}}
\def\zz{\mathbf{z}}

\def\eps{\varepsilon}
\def\R{{\mathbb R}}

\def\h_j{{\mathcal h_j}}

\def\h_j{\hspace{5 pt}}

\def\eps{\varepsilon}
\def\R{{\mathbb R}}

\def\h_j{{\mathcal h_j}}

\def\argmin{\mathop{{\rm argmin}}\nolimits}

\def\spt{\mathop{{\rm spt}}\nolimits}

\def\Tr{\mathop{{\rm Tr}}\nolimits}

\def\dv{\mathop{{\rm div}}\nolimits}

\def\sym{\mathop{{\rm sym}}\nolimits}

\def\capacity{\mathop{{\rm cap}}\nolimits}
\def\Capacity{\mathop{{\rm Cap}}\nolimits}
\def\u{\mathbf{u}}
\def\v{\mathbf{v}}

\def\v{{\bf v}}
\def\w{{\bf w}}

\def\x{{\bf x}}

\def\wconv{\rightharpoonup}

 \newenvironment{proofad1}{\removelastskip\par\medskip
\noindent{\textbf {Proof of Theorem \ref{mainth1comp}}.}
\rm}{\penalty-20\null\hfill$\square$\par\medbreak} 

\renewcommand{\epsilon}{\varepsilon}
\newcommand{\beeq}{\begin{equation}}
\newcommand{\eneq}{\end{equation}}
\newcommand{\bear}{\begin{array}}
\newcommand{\enar}{\end{array}}
\newcommand{\bema}{\begin{displaymath}}
\newcommand{\enma}{\end{displaymath}}

\newcommand{\beea}{\begin{eqnarray}}
\newcommand{\enea}{\end{eqnarray}}

\newcommand{\om}{\Omega}

\newcommand{\lab}[1]{ \label{#1} }

\def\wconv{\rightharpoonup}

\usepackage{graphicx}        
\usepackage{tikz}
 \usetikzlibrary{arrows}
 \usetikzlibrary{calc}
\usepackage{pgfplots}
\usetikzlibrary{intersections}

\title[]{Limit of nonlinear Signorini problems for incompressible materials}
   \author[Edoardo Mainini, Danilo Percivale, Robertus van der Putten]{Edoardo Mainini, Danilo Percivale, Robertus van der Putten}
 \address{}
 \email{}
 \address{Universit\`{a} degli Studi di Genova, Dipartimento di   Ingegneria Meccanica, Energetica, Gestionale e dei Trasporti (DIME),
  Via Opera Pia 15, 16129 Genova, Italy}
  \email{edoardo.mainini@unige.it}
  \email{danilo.percivale@unige.it}
  \email{robertus.van.der.putten@unige.it}
\begin{document}
 \maketitle
\begin{abstract} The paper is devoted to the linearization of the non linear Signorini functional in the incompressible case. The limit functional, in the sense of $\Gamma$-convergence, may coincide with the expected one only in some particular cases. 
\end{abstract}

\begin{flushleft}
  {\bf AMS Classification Numbers : 49J45 · 74K30 · 74K35 · 74R10\,} 
\end{flushleft}
\vskip0.5cm
\maketitle

\section{Introduction} 
{In linear elastostatics, the classical Signorini problem \cite{Si} requires to find the  equilibrium of an elastic body subject to external forces and  resting on a rigid support $E\!\subset\! \partial\om$ in its reference configuration $\om$.  Precisely, if $\om$ is subject to suitable volume and surface forces ${\mathbf f}:\Omega\rightarrow \R^3$ and  ${\mathbf g}:\partial\Omega\!\setminus\! E\rightarrow \R^3$ such that
\begin{equation} \lab{load}\mathcal L(\u)\,:=\,\int_\om\mathbf f\cdot\u\ d\x+\int_{\partial\Omega\setminus E}\mathbf g\cdot \u\ d{\mathcal H}^2
\end{equation}
is the load potential and $\mathbf u:\Omega\to\mathbb R^3$ is the displacement field, then, by assuming that $\mathcal H^2(E)\! >\! 0$, the variational formulation of the Signorini problem consists in finding a minimizer of the functional 
\begin{equation}\lab{linel}
\mathcal E(\u):= \displaystyle\int_\Omega\mathcal Q\big(\x,\mathbb E(\u)\big)\,d\x -\mathcal L(\u)
\end{equation}
among all $\u$ in the Sobolev space $ H^1(\om;\mathbb R^3)$ such that $\u\cdot\mathbf n\ge 0$ $\,\mathcal H^2$-\! a.e. on $E$ where $\mathbf n$ is the inward unit vector normal to $\partial\om$  and $\mathcal H^2$ is the two-dimensional Hausdorff measure. As usually happens, here
$ \mathbb E$  denotes the linear strain tensor, $\mathbb C$ represents the classical linear elasticity tensor and $$
\mathcal Q(\x,\mathbb E)\,:=\, \frac{1}{2}\,\mathbb E^T\mathbb C(\x)\,\mathbb E
$$  
 is the corresponding strain energy density (see \cite{Gu}).
A classical result (see \cite{F2}) states that a minimizer  exists  if $\mathcal L(\v)\!\le\! 0\,$ for
every infinitesimal rigid displacement $\v$ such that $\,\v\cdot\n\!\ge\! 0$ $\,\mathcal H^2$-\! a.e. on $E$ and  $\mathcal L(\v) \!=\! 0$ if and only if $\v\cdot\n \equiv 0$ $\,\mathcal H^2$-\! a.e. on $E$.
More recently a proper generalization of the latter formulation has been given by assuming that the set $E$ has positive Sobolev capacity 
and accordingly modifying the obstacle condition by requiring $\v\cdot\n\ge 0$ on $E$ up to a set of null capacity (shortly, q.e. on $E$): the existence of minimizers for this general setting was proved in \cite[Theorem 4.5]{BBGT}.
Although the original  formulation given in \cite{Si}  may look different from the generalized notion exploited in \cite{BBGT}, it can be shown (see Remark \ref{2.4}) that if the set $E\!\subset\!\partial\Omega$ is regular in an appropriate sense  then the two frameworks coincide.\\

In the  recent  paper \cite{MPTOB} it has been shown that, under sharp conditions on $\mathcal L$, there exists a sequence of (suitably rescaled) functionals $\mathcal G_h$ of finite elasticity such that $\inf \mathcal G_h$ converges to the minimum of $\mathcal E$ and that  in addition there are  examples in which this convergence fails.
The aim of this paper is to show that at least in the planar obstacle case, namely when $E\subset \partial\om\cap\{\x\cdot\mathbf e_3=0\}\not\equiv \emptyset$ (being $(\mathbf e_i),$ $i=1,2,3,$  the canonical basis of $\R^3$)  a similar result holds also for incompressible energy densities.
More in detail,  denoting by $\mathbf y:\Omega\to\mathbb R^3$  the deformation field and by $h> 0$ an adimensional  parameter,  we introduce \MMM a strain  energy density $\mathcal W:\Omega\times\mathbb R^{3\times3}\to [0,+\infty]$, which is frame indifferent and minimized at the identity, and \KKK the family of energy functionals 
\begin{equation*}
\displaystyle \mathcal F_h^I(\mathbf y):=h^{-2}\int_\om\mathcal W^I\big(\x,\nabla\mathbf y(\xx)\big)\,d\x
-h^{-1}\mathcal L(\mathbf y-\mathbf x)
\end{equation*}
where  $\mathcal L$ is defined as in \eqref{load} and $\mathcal W^I$ is the incompressible strain energy density (\MMM coinciding with $\mathcal W$ if  $\det \nabla \yy\equiv 1$ and set equal to $+\infty$ otherwise\KKK).
We define \[
\lab{Gh} \mathcal G_h^I(\mathbf y)=\left\{\begin{array}{ll}  \mathcal F_h^I(\mathbf y)& \ \ \hbox{if }\ 
{\mathbf y}\in \mathcal A
\vspace{0.1cm}\\
 \ +\infty &\ \ \hbox{otherwise in $H^1(\om;\R^3)$},\\
\end{array}\right.
\]
where $E$, the portion of the elastic body that is sensitive to the obstacle, is such that $\capacity E > 0$, and where $\mathcal A=\mathcal A(E)$ denotes the class of admissible deformations (i.e., those $\yy\in H^1(\om;\R^3)$ such that $\mathbf y\cdot\mathbf e_3\ge 0$ quasi-everywhere on $E$, i.e., up to sets of null capacity\KKK).
 Moreover we have to assume that 
 $
\mathcal L(\mathbf y-\mathbf x)\le 0
$
for every \MMM rigid deformation $\mathbf y\in \mathcal A$. Indeed, if there exists a rigid deformation $\overline\yy\in\mathcal A$ such that $\mathcal L(\overline\yy-\x)>0$, then   \KKK
 $$ \mathcal G_h^I(\overline\yy)=-h^{-1}\mathcal L(\overline\yy-\mathbf x)\to -\infty\ \text{as}\ h\to 0^+.$$
 Therefore, since a rigid deformation $\mathbf y(\x)=\mathbf R \x+\mathbf c$ with $\mathbf R\!\in\! SO(3),\ \mathbf c=(c_1,c_2,c_3)\in \mathbb R^3,$ belongs to $\mathcal A$ whenever  $c_3\ge  -\left(\mathbf R\x\right)_3$ on $E$, 
 we have to assume
 \begin{equation}\label{rc}
 \mathcal L( (\mathbf R - \mathbf I) \x+\mathbf c)\le 0\quad\mbox{for every $\mathbf R\in SO(3)$ and  every $\mathbf c\in \mathbb R^3$ s.t. $c_3\ge  -\left(\mathbf R\x\right)_3$ on $E$.}
 \end{equation}
In the main result of this  paper (Theorem \ref{mainth1comp} below) we show that if $\mathcal L$ satisfies the  necessary condition \eqref{rc}  together with $\mathcal L(\mathbf e_3) < 0$ and  \beeq\lab{shear} \mathcal L\left (\,(\mathbf R\x-\x)_1\,\mathbf e_{1}+ (\mathbf R\x-\x)_2\,\mathbf e_{2}\,\right )\le 0\quad  \forall\,\mathbf R\in SO(3),\eneq
and if $\mathcal W$ satisfies some standard assumptions to be detailed in the next section, then
{\color {blue}}
\[\displaystyle\lim_{h\to 0}\inf_{H^1(\om;\R^3)} \mathcal G_h^I=\min_{H^1(\om;\R^3)} \mathcal G^I,\]
where, with the notation $D^2$ for Hessian of $\mathcal W(\x,\cdot)$, the limit functional is characterized as
\[
\label{elfunc}
\displaystyle {\mathcal G}^I(\u):=
\left\{\begin{array}{ll}
\displaystyle \int_\om \mathcal Q^I(\x, \mathbb E(\u))\,d\x-\max_{\mathbf R\in{\mathcal S}_{\mathcal L, E}}{\mathcal L}(\mathbf R \u)
\quad 
& \ \ \hbox{if }\ 
\u\cdot\mathbf e_{3}\ge 0\ \hbox{quasi-everywhere on}\ E 
\vspace{0.1cm}
\\ 
 \ \ +\infty\quad &\ \ \hbox{otherwise in $H^1(\om;\R^3)$ }
\end{array}
\right.
\]
\begin{equation*}
\mathcal Q^I(\x,\mathbf F):=\left\{\begin{array}{ll}\vspace{0.2cm}
 \dfrac12\,\mathbf F^T\,D^2 {\mathcal W}(\x,\mathbf I)\,\mathbf F\qquad &\hbox{if }\; \tr\,\mathbf F=0\\
+\infty\ &\hbox{otherwise,}\\
\end{array}\right.
 \end{equation*}
\begin{equation*}
\label{kernel}
 \mathcal S_{\mathcal L,E}\,=\,\left\{\, \mathbf R\in SO(3): \ \mathcal L\left ( (\mathbf R- \mathbf I)\,\x\right )\,-\inf_{\xx\in E}\big(\mathbf (\mathbf R\xx)_3\big)\mathcal L(\mathbf e_3 )=0 \, \right\}.
\end{equation*}
\MMM
In general, $\mathcal G^I$ differs from the expected limit functional $\mathcal E^I$,  defined by replacing $\mathcal Q$ with $\mathcal Q^I$ in \eqref{linel}.
However,  under the hypotheses on $E$ and $\mathcal L$ that were previously detailed,  it has been proven in \cite{MPTOB} (see Lemma \ref{lemma alternative} below) that either 
$ \mathcal S_{\mathcal L ,E}\equiv \{\mathbf I\}$ or
 $\mathcal S_{\mathcal L,E}\,=\,
\left\{
\mathbf R\in SO(3): 
\mathbf R\mathbf e_3=\mathbf e_3
\right\}.$
 If $ \mathcal S_{\mathcal L ,E}\!\equiv\! \{\mathbf I\}$ then clearly 
${\mathcal G}^I\!\equiv\! \mathcal E^I$, hence in this case we recover the minimum of the Signorini problem in linearized elasticity as the limit of the nonlinear energies $\inf \mathcal G_h^I$. 
In particular, if $\Omega$ is  
contained in the upper half-space, $E =\partial\Omega\cap\{x_3=0\}$,  if  $\mathbf f= f\mathbf e_3,\ \mathbf g= g\mathbf e_3$ and  $\mathcal L$
satisfies conditions {\eqref{rc}}-\eqref{shear} and $\mathcal L(\mathbf e_3)<0,$ then since  $\mathcal L(\v)=\mathcal L(\mathbf R\v)$ for every $\mathbf R\in \mathcal S_{\mathcal L, E}$ we get  $\min \mathcal G^I\!=\!\min \mathcal E^I$ hence $\inf \mathcal G_h^I\to \min\mathcal E^I$ in this case.}\\
 
 \MMM
 \textbf{Plan of the paper.} In section \ref{sect Main Results} we rigorously provide assumptions on the obstacle, on the external forces and the global energy functionals. Then, we state the main variational convergence result. In Section \ref{proofssection} we give the proof of the main result after having
 proved suitable technical lemmas. With respect to the analysis from \cite{MPTOB} about the compressible case, some care is needed in constructing the suitable upper and lower bounds while simultaneously taking care of the obstacle condition and the incompressibility constraint.
 \KKK
 
\section{Notation and main results} 
\label{sect Main Results}

In the following, $\Omega$ will denote the reference configuration of an elastic body. $\Omega$  is always assumed to be a nonempty, bounded,  connected, Lipschitz  open  set in $\mathbb R^3$.

\MMM

\subsection{The obstacle} 
Notations $\xx=(x_1,x_2,x_3)$ and $\yy=(y_1,y_2,y_3)$ will be used to represent generic points in $\R^3$, with components referred to the canonical basis  
$(\mathbf e_i)$, $i=1,2,3$.
In the Signorini problem,  the elastic body rests on a frictionless rigid support $E\subset\partial\Omega$. The set $E$ will be assumed to be planar, i.e., contained in $\{x_3=0\}$, and of positive capacity, according to the next definition. 

\MMM  \KKK
\MMM
For every compact set $K\subset\mathbb R^N$  we define the Sobolev capacity of $K$ by setting, see \cite[Section 2.2]{AH},
\begin{equation*}
\capacity K\!=\! \inf \left\{\|w\|_{H^1(\R^N)}^2: w\!\in\! C^{\infty}_0(\mathbb R^N),\, w\ge 1\,\text{on}\, K\right\}.
\end{equation*}
If $G\subset \R^N$ is open we define 
\[
\capacity G:= \sup\{\capacity K: K\ \text{compact},\ K\subset G\},
\]
and for a generic set $F\subset\mathbb R^3$
\[
\capacity F:=\inf\{\capacity G: G\ \text{open},\ F\subset G\}.
\]
%
\KKK
A property is said to hold quasi-everywhere (q.e. for short) if it holds true outside a set of zero capacity. 
We introduce (see \cite{BBGT}) a canonical representative of a set $F$, called the \textsl{essential part} of $F$ and denoted by $F_{ess}$, which coincides with $F$  whenever $F$  is a smooth closed manifold  or the closure of an open subset of $\mathbb R^N$:
 for every set $F\subset\R^3$ we define
 \begin{equation*}
 F_{ess}\ :=\ \bigcap \,\{\, C:\ C \hbox{ is closed and }  \capacity(F\!\setminus\! C)=0\,\}.
\end{equation*}
As shown in \cite{BBGT}, we have
\[  F_{ess}\ \text{ is a closed subset of}\  \overline F,\qquad \capacity(F\!\setminus\! F_{ess})=0,\qquad \capacity F=0\ \text{ if and only if}\  F_{ess}=\emptyset.\]
\MMM Moreover, if $\capacity F>0$ and $\x\in F_{ess}$, then $\capacity(F\cap B_r(\x))>0$ for every $r>0$.\\ \KKK


Throughout the paper we will assume that \MMM the set $E$ giving the obstacle condition satisfies \KKK
\begin{equation}
\lab{H2 on E}
E\subset \partial\om\cap \{x_3=0\}\not\equiv \emptyset\qquad\mbox{and}\qquad \capacity E>0 .
\end{equation}

We recall that every function in  $H^{1}(\Omega;\R^3)$  has a precise representative defined quasi-everywhere on the whole $\overline\Omega$.
Indeed, if $\uu\!\in\! H^{1}(\Omega;\R^3)$ and $\vv\!\in\! H^{1}(\R^3;\R^3)$ is a Sobolev extension of $\uu$, it is well known (see \cite[Proposition 6.1.3]{AH}) that the limit 
\begin{equation*}
\displaystyle
 {\vv}^* (\xx)\ :=\ \lim_{r \downarrow 0}\ 
 \frac{1}{|B_r(\xx)|}\,\int_{ B_r(\xx)}\! \vv(\mathbf \xi)\,d\mathbf \xi
 \end{equation*}
exists for q.e.\,$\x\!\in\!\R^3.$ The function ${\vv}^* $ is called the precise representative of $\vv$ and is a quasicontinuous function in $\R^3$, meaning that  for every $\eps> 0$ there exists an open set $V\subset \R^3$ such that $\capacity V < \eps$ and
${\vv}^*$ is continuous in $\R^3\!\setminus\! V$.
It has been proven in \cite{MPTOB} that if $\vv_1,\ \vv_2$ are two distinct Sobolev extensions of $\uu$ then
$
\vv_1^*(\xx)\ =\ \vv_2^*(\xx)$
 for {q.e.} $\xx\in \overline\om$. 
 Therefore if $\uu\in H^{1}(\Omega;\R^3)$ we may define its precise representative for quasi-every $\xx\in \overline\om$ 
by 
 \beeq\lab{prec}{\uu}^* (\xx)=\lim_{r \downarrow 0}\ 
 \frac{1}{|B_r(\xx)|}\,\int_{ B_r(\xx)}\! \vv(\mathbf \xi)\,d\mathbf \xi,
 \qquad
 \text{q.e. }\xx\in \overline\om,
\eneq
where $\vv$ is any Sobolev extension of $\uu$. 
The function ${\uu}^* $ is 
pointwise quasi-everywhere defined by \eqref{prec} and
is \textsl{quasicontinuous} on $\overline\om$ i.e., for every $\eps> 0$ there exists a relatively  open set $V\subset \overline\om$ such that $\capacity V < \eps$ and
${\uu}^*$ is continuous in $\overline\om\setminus V$.

\begin{remark}\lab{2.4}\MMM \rm
If $w\in H^1(\Omega)$ then its negative part $w^-$ is in $H^1(\Omega)$ too.
Moreover, both $(w^-)^*$ and  $(w^*)^-$ are quasicontinuous in $\overline\om$ and $(w^-)^*=(w^*)^-=w^-$ a.e. in $\om$. Then, by \cite{K}, $(w^-)^*=(w^*)^-$ q.e. in $\overline\om$.
Therefore the condition $(w^-)^*=0$ q.e. in $E_{ess}$ is equivalent to $w^*\ge 0$ q.e. in $E_{ess}$.
In particular, 
as it was pointed out in  in \cite{MPTOB}, Theorem 2.1 of \cite{ET} entails that if $E_{ess}\subset \partial\om$ is Ahlfors $2$-regular  and $\mathcal H^2(E_{ess}) > 0$, then the condition $w\ge 0\ $ q.e. on $E_{ess}$ is equivalent to $w\ge 0\ \ \mathcal H^2$ a.e. on $E_{ess}$ so the classical framework of 
\cite{BBGT,Si,KS} 
 is equivalent to ours in this case.\KKK
\end{remark}


%
\subsection{The elastic energy density} 
Let $\mathbb R^{3\times 3}$ denote the set of $3\times 3$ real matrices, endowed with the Euclidean norm $|\mathbf F|=\sqrt{\tr( \mathbf F^T\mathbf F)}$. We let  ${\mathrm{sym} \mathbf F}:=\tfrac12(\mathbf F^T+\mathbf F)$.  $SO(3)$ will  denote the special orthogonal group.
Let ${\mathcal L}^3$ and ${\mathcal B}^3$ denote respectively the
$\sigma\mbox{-algebras}$
of Lebesgue measurable and Borel measurable subsets of $\R^3$
and let $\mathcal W : \om \times \mathbb R^{3 \times 3} \to [0, +\infty ]$ be  ${\mathcal L}^3\! \times\! {\mathcal B}^{9} $- measurable 
 satisfying the following standard assumptions, see also \cite{MaiPe2,MaiPe3}: 
%
%
\beeq \lab{framind}\W(\x, \mathbf R\mathbf F)=\W(\x, \mathbf F) \quad \ \forall \ \mathbf\! \mathbf R\!\in\! SO(3) \ \ \forall\ \mathbf F\!\in\! \mathbb R^{3 \times 3},\quad\qquad \ \mbox{for a.e. $\x\!\in\!\Omega$},
\eneq
\beeq \lab{Z1}
\min_{\mathbf F} \W(\x,\mathbf F)= \W(\x,\mathbf I)=0 \quad \mbox{for a.e. $\x\in\Omega$}
\eneq
and as far as it concerns the regularity of $\mathcal W$, we assume that there exist an open neighborhood $\mathcal U$ of $SO(3)$ in $\R^{3\times3}$,  an increasing function $\omega:\mathbb R_+\to\mathbb R$ satisfying $\lim_{t\to0^+}\omega(t)=0$ and a constant $K>0$
such that for a.e. $\x\in\om$
\beeq\begin{array}{ll}\lab{reg}&   
\vspace{0,1cm}
\mathcal W(\x,\cdot)\in C^{2}(\mathcal U),\;\;\;
 \vspace{0,1cm}
  |D^2 \mathcal W(\x,\mathbf I)|\le K \;\;\hbox{and}\\
& 
 |D^2\W(\x,\mathbf F)-D^2\W(\x,\mathbf G)|\le\omega(|\mathbf F-\mathbf G|),\quad\forall\; \mathbf F,\mathbf G\in\mathcal U.
\end{array}
\eneq
Moreover we assume that there exists $C>0$  such that for a.e. $\x\in\Omega$
\beeq 
\lab{coerc}
\begin{array}{ll}
\W(\x,\mathbf F)\ge  C (d(\mathbf F, SO(3)))^2\qquad
 \forall\, \mathbf F\in \mathbb R^{3 \times 3},
\end{array}
\eneq
where $d(\,\cdot\,,SO(3))$ denotes the Euclidean distance function from the set of rotations.
In order to consider incompressible elasticity models, starting from a function $\mathcal W$ as above we  introduce the incompressible strain energy density by letting, for a.e. $\x\in\Omega$,
\begin{equation}\lab{winc}
{\mathcal W}^I (\x, \mathbf F):=\left\{\begin{array}{ll} \mathcal W(\x,\mathbf F) \qquad &\hbox{if}\ \ \det\mathbf F=1\\
 +\infty\ &\hbox{otherwise.}
\end{array}\right.
\end{equation} \MMM

We recall some basic properties that follow from the above assumptions, see \cite{MaiPe2, MaiPe3}. For a.e. $\x\in\om$  there holds 
%
%
\[\displaystyle 
\lim_{h\to 0} h^{-2}\mathcal W(\x,\mathbf I+h\mathbf F)=\frac12\, \mathbf F^T D^2\mathcal W(\x,\mathbf I)\,\mathbf F=\frac12
\,{\rm sym} \mathbf  F\, D^2\mathcal W (\x, \mathbf I) \ {\rm sym}\mathbf  F\qquad \ \forall\ \mathbf F\!\in\! \mathbb R^{3 \times 3},
\]
hinting to the linear elastic energy density as limit of the nonlinear energies, with $\mathbb C(\x)= D^2\mathcal W (\x, \mathbf I)$,  and 
\begin{equation*}\label{ellipticity}\begin{aligned}
\frac12\, \mathbf F^T D^2\mathcal W(\x,\mathbf I)\,\mathbf F&
\ge 
C|\mathrm{sym}\mathbf F|^2\qquad \forall\ \mathbf F\!\in\! \mathbb R^{3 \times 3},\end{aligned}
\end{equation*}
where $C$ is the constant in \eqref{coerc}.
%
 Taking  \eqref{reg} into account we have, for a.e. $\x\in\om$,
\begin{equation}\lab{regW}
\left|\mathcal W(\x, \mathbf I+h\mathbf F)- \frac{h^2}{2} \,{\rm sym} \mathbf  F\, D^2\mathcal W (\x, \mathbf I) \ {\rm sym}\mathbf  F\right|\le h^2\omega(h|\mathbf F|)|\mathbf F|^2
\end{equation}
for every $\mathbf F\in\R^{3\times3}$ and every $h>0$ such that $\mathbf I+h\mathbf F\in\mathcal U$. 

We mention a class of energy densities $\W$ (the so called Yeoh materials \cite{Y1,Y2})  fulfilling the assumptions  above \eqref{framind}--\eqref{coerc} and for which the main result of the present paper (see Theorem \ref{mainth1comp} below) applies 
(for simplicity, we consider the homogeneous case): 
\begin{equation*}
\W(\mathbf{F}):=\sum_{k=1}^3 c_k(|\mathbf{F}|^2-3)^k
\end{equation*}
with $c_k> 0$.
It is easy to check that with this choice the  energy density  satisfies all assumptions from \eqref{framind} to \eqref{reg} while inequality  \eqref{coerc} has been proven in \cite{MOP}.  \KKK
It is worth noticing that when material constants are suitably chosen then also the classical Ogden-type energies may fulfill the assumptions from \eqref{framind} to \eqref{coerc} and we refer to  \cite{MOP} for all details.
\subsection{External forces}
We introduce  a  body force field $\mathbf f\in L^{6/5}(\Omega,\R^3)$ and a surface force field $\mathbf g\in L^{4/3}(\partial\Omega,\R^3)$. From now on,  $\mathbf f$ and $\mathbf g$ will  always be understood to satisfy these summability assumptions. The load functional is the following linear functional \begin{equation*}
\mathcal L(\v):=\int_\om\mathbf f\cdot\v\,d\x+\int_{\partial\om}\mathbf g\cdot\v\,d\mathcal H^2(\x),\qquad \v\in H^1(\om,\mathbb R^3).
\end{equation*}
We note that since $\Omega$ is a bounded Lipschitz domain, the Sobolev embedding $H^{1}(\Omega,\mathbb R^3)\hookrightarrow L^{6}(\Omega,\mathbb R^3)$ and the Sobolev trace embedding $H^{1}(\Omega,\mathbb R^3)\hookrightarrow L^{4}(\partial\Omega,\mathbb R^3)$ imply that $\mathcal L$ is a bounded functional over $H^{1}(\Omega,\mathbb R^3)$ and throughout the paper we denote its norm with $\|\mathcal L\|_*$.

For every $\mathbf R\!\in\! SO(3)$ we set
\[ \displaystyle\mathcal C_{\mathbf R}:=\{\mathbf c: c_3\ge -\min_{\x\in E_{ess}}(\mathbf R\x)_3\}\]
and we assume the following \textsl{geometrical compatibility between load and obstacle}
\beeq\lab{L0} \mathcal L((\mathbf R-\mathbf I)\x+\mathbf c)\le 0\qquad \forall \mathbf R\in SO(3),\ \forall \mathbf c\,\in \mathcal C_{\mathbf R}\eneq
together with 
\beeq\lab{shear2} \mathcal L\left ((\mathbf R\x-\x\right)_\alpha\mathbf e_\alpha)\le 0\quad  \forall\,\mathbf R\in SO(3),\eneq
the summation convention over the repeated index $\alpha=1,2$ being understood all along this paper.
It can be shown that condition \eqref{L0} is equivalent to (see \cite[Remark 3.4]{MPTOB})
\beeq
\lab{L1} 
\mathcal L(\mathbf e_3)\le 0=\mathcal L(\mathbf e_1)=\mathcal L(\mathbf e_2) ,
\qquad\quad
\Phi(\mathbf R, E, \mathcal L)\le 0
\quad \forall\,\mathbf R\!\in\! SO(3) ,
\eneq
where we have set
$$\Phi(\mathbf R, E, \mathcal L):=\mathcal L((\mathbf R-\mathbf I)\x)\,-\,\mathcal L(\mathbf e_3)\min_{\x\in E_{ess}}(\mathbf R\x)_3$$
and from now on we will use \eqref{L1} in place of \eqref{L0}. It has been shown in \cite{MPTOB} that conditions \eqref{shear2} and \eqref{L1} are compatible and that they do not imply each other, see also Remark \ref{shear3}  below.
We next state three auxiliary lemmas, in order to better clarify the role 
of conditions \eqref{shear2} and \eqref{L1}. Proofs may be found in \cite{MPTOB}. In the following, $\wedge$ denotes the cross product.
\begin{lemma}\lab{load0}
Assume that \eqref{shear2} holds. Then, 
\beeq\lab{equivshear} \mathcal L ((\mathbf a\wedge \x)_{\alpha}\,\mathbf e_\alpha)=0\ \  \text{and}\ \ \mathcal L ((\mathbf a\wedge  (\mathbf a\wedge  \x))_{\alpha}\,\mathbf e_\alpha)\le 0\ \ \ \forall\, \mathbf a\in \mathbb R^3.
\eneq
\end{lemma}
\begin{remark}\lab{0L}\rm It is worth noticing that, by  inserting $\mathbf a=\mathbf e_1$ or $\mathbf a=\mathbf e_2$, the condition \eqref{equivshear} entails
$\mathcal L(x_3\mathbf e_2)=0$ and $\mathcal L(x_3\mathbf e_1)=0$ respectively.
\end{remark}
\begin{lemma}\label{load1}
 Assume 
 \eqref{H2 on E} and 
\eqref{L0}. Then
\begin{enumerate}
 \item[(1)] 
   $\, \mathcal L(\mathbf e_1) = \mathcal L(\mathbf e_2) = 0$ and $
     \mathcal L(\mathbf e_3) \le 0   ;  $
 \vskip0.05cm
  \item[(2)] 
   $\, \mathcal L(\mathbf e_3\wedge \xx) = 0  ; $
  \vskip0.05cm
  \item[(3)] 
   $\, \mathcal L\big(\mathbf e_3\wedge (\mathbf e_3\wedge\xx)\big) \le 0  ; $
   \vskip0.04cm
 \item[(4)]   there exists $\xx_{\mathcal L}$ in the relative interior of the convex envelope of  $E_{ess} $ such that\\ $\ \mathcal L\big(\,\mathbf a \wedge (\xx-\xx_{\mathcal L})\big)=0$ $\ \forall\mathbf a\in\R^3$.
\end{enumerate}
\end{lemma}
\begin{remark}\rm
We emphasize that conditions (1) and (4) in  Lemma \ref{load1} together with \eqref{H2 on E} and $\mathcal L(\mathbf e_3)< 0$ coincide with conditions (4.9)--(4.11) of  \cite[Theorem 4.5]{BBGT}, which provides the solution to Signorini problem in linear elasticity.
We also emphasize that
the whole set of conditions (1)--(4) appearing in the claim of Lemma \ref{load1} together with 
condition~\eqref{H2 on E} on the set $E$ is not equivalent to admissibility of the loads as expressed by \eqref{L1}. This phenomenon is made explicit in \cite[Example 3.6]{MPTOB}.\end{remark}
Finally, by setting 
\beeq
\label{kernel}
 \mathcal S_{\mathcal L,E}=\left\{\, \mathbf R\in SO(3): \Phi(\mathbf R, E, \mathcal L)=0 \,\right\},
\eneq
we have
\begin{lemma}\label{lemma alternative}
Assume that  \eqref{H2 on E}, \eqref{L1} hold and that $\mathcal L(\mathbf e_3) < 0$.
Then
\begin{equation*}
 \hbox{either}\quad
 \mathcal S_{\mathcal L,E}=\{\,\mathbf I\,\}
 \qquad
 \hbox{or}
 \quad
  \mathcal S_{\mathcal L,E}=\{\,\mathbf R\in SO(3): \mathbf R\mathbf e_3=\mathbf e_3\,\}.
\end{equation*}
\end{lemma}

\subsection{Energy functionals} 
If $E$ fulfils \eqref{H2 on E}
the incompressible Signorini problem in linear elasticity
can be described as the minimization of the functional 
\ $\mathcal E^I: H^1(\Omega,\mathbb R^3)\to\mathbb R\cup\{+\infty\}$\, defined by
\[
\displaystyle {\mathcal E}^I(\u):=\left\{\begin{array}{ll}\displaystyle \int_\om \mathcal Q^I(\x, \mathbb E(\u))\,d\x-\mathcal L(\u)\quad &\hbox{if} \ \u\in \mathcal A
\vspace{0.1cm}
\\
 \ \!\!+\infty\quad &\hbox{otherwise in} \ H^1(\om,\mathbb R^3),
\end{array}\right.
\]
where $\mathbb E(\u):=\! \sym \nabla\u$,  $\;\;\mathcal Q(\x, \mathbf F):=\tfrac12\,\mathbf F^T\mathbb C(\x)\,\mathbf F$,  $\;\;\mathbb C(\x):=D^2\mathcal W(\x,\mathbf I)$,
\[\lab{QI}{\mathcal Q}^I (\x, \mathbf F):=\left\{\begin{array}{ll} &\mathcal Q(\x, \mathbf F)\qquad \hbox{if}\;\;\tr\,\mathbf F=0\\
&\\
& +\infty\ \hbox{otherwise}\\
\end{array}\right.
\]
  and where $\mathcal A$ is defined by
\begin{equation*}
 \mathcal A\ :=\ 
 \left\{
   \,\uu\in H^1(\Omega;\R^3):\ \,  u_3^*(\xx)\ge 0 \ \,\hbox{q.e.}\ \xx\!\in\! E \,\right\} .
\end{equation*}
The meaning of such constraint is that 
the deformed configuration of $E$, namely $\{\mathbf y(\xx)\!:=\xx+\u(\xx),\, \xx\!\in\! E\}$, is constrained to remain in $\{y_3\ge 0\}$.

For every $\mathbf y\in H^1(\om,\mathbb R^3)$ we introduce the set 
\begin{equation}\lab{roty}
\mathcal M(\mathbf y)\,:= \,
\argmin\left\{ \,\int_\om |\nabla\mathbf y-\mathbf R|^2\,d\x: 
\ \mathbf R\in SO(3)\right\} .
\end{equation}
Thus, 
due to the rigidity inequality of \cite{FJM0}, there exists a constant $C=C(\om) >0$ such that for every $\yy\in H^1(\om,\mathbb R^{3})$ and every $\mathbf R\in \mathcal M(\mathbf y)$
\begin{equation}\label{muller}
 \int_{\om}\Big(d\big(\nabla \yy, SO(3)\big)\Big)^2\,d\x
 \,\ge\,C\int_{\om}|\nabla\yy-\mathbf R|^2\,d\x,
\end{equation}
where $d\big(\mathbf F,SO(3)\big):=
\min\{|\mathbf F-\mathbf R|:\mathbf R\in SO(3)\}$. \MMM The rigidity inequality is a crucial tool for obtaining linear elasticity as limit of finite elasticity via $\Gamma$-convergence, as seen in \cite{DMNP} and subsequent works \cite{JS, KM, MPTJOTA, traction, MPTOB, MOP, MaiPe1, MaiPe2, MaiPe3, MaorMora, MoraRiva}. \KKK

\MMM Let us introduce nonlinear elastic energies. If $(h_j)_{j\in\mathbb N}\subset(0,1)$ is a vanishing sequence, 
 the rescaled finite elasticity functionals  $\mathcal G_j: H^1(\Omega,\R^3)\to \R\cup\{+\infty\}$ are defined  by  
\[ 
\mathcal G_j(\mathbf y)=\left\{\begin{array}{lr} & \displaystyle h_j^{-2}\int_\om\mathcal W(\x,\nabla\mathbf y)\,d\x
-h_j^{-1}\mathcal L(\mathbf y-\mathbf x)\ \qquad  \hbox{if}\ \mathbf y\in \mathcal A\\
&\\
& +\infty \ \qquad    \hbox{otherwise.}\\
\end{array}\right.
\]\KKK
The related rescaled incompressible finite elasticity functionals  $\mathcal G_j^I: H^1(\Omega,\R^3)\to \R\cup\{+\infty\}$ are defined by  
\[
\mathcal G_j^I(\mathbf y)=\left\{\begin{array}{lr} & \displaystyle h_j^{-2}\int_\om\mathcal W^I(\x,\nabla\mathbf y)\,d\x
-h_j^{-1}\mathcal L(\mathbf y-\mathbf x)\ \qquad  \hbox{if}\ \mathbf y\in \mathcal A,\\
&\\
& +\infty \ \qquad    \hbox{otherwise.} \\
\end{array}\right.
\]


 It will be shown (see Lemma \ref{lemma compactness} below) that $\inf_{H^1(\om;\R^3)} \mathcal G_j^I> -\infty$ for every $j\in \mathbb N$. Moreover,  $(\mathbf y_j)_{j\in\mathbb N}\!\subset\! H^1(\Omega,\mathbb R^3)$ is said a \textsl{minimizing sequence of the sequence of functionals} $\mathcal G_j^I$ if 
\beeq\lab{qmin}
 \lim_{j\to+\infty}\left(\,\mathcal G_j^I(\mathbf y_j)\,-\inf_{H^{1}(\Omega,\mathbb R^3)}\mathcal G_j^I\right)=0. 
 \eneq
The main focus of the paper is to investigate if  $\mathcal G_j^I(\yy_j)$ converges to a minimum of \eqref{qmin} whenever  $(\yy_j)_{j\in\mathbb N}$ is a minimizing sequence of $\mathcal G_j^I$.
To this end we introduce 
the functionals \MMM $\mathcal I^I,\ \widetilde{\mathcal G}^I,\ \mathcal G^I  :H^1(\Omega,\mathbb R^3)\to\mathbb R\cup\{+\infty\}$  \KKK defined by
\[
\mathcal I^I(\uu):= \displaystyle \min_{\mathbf b\in \mathbb R^2}\int_\om \mathcal Q^I(\x, \mathbb E(\u)+\tfrac{1}{2}b_{\alpha}(\mathbf e_{\alpha}\otimes\mathbf e_3+\mathbf e_{3}\otimes\mathbf e_\alpha))\,d\x,
\]
\[
\displaystyle {\widetilde{\mathcal G}}^I(\u):=
\left\{\begin{array}{ll}
\displaystyle \mathcal I^I(\uu)-\max_{\mathbf R\in \mathcal S_{\mathcal L, E}} \mathcal L(\mathbf R\u)\quad &\hbox{if} \ \u\in \mathcal A
\vspace{0.1cm}
\\
 \ \!\!+\infty\quad &\hbox{otherwise in} \ H^1(\om,\mathbb R^3),
\end{array}\right.
\]
and
\[
\displaystyle {\mathcal G}^I(\u):=
\left\{\begin{array}{ll}
\displaystyle \int_\om \mathcal Q^I(\x, \mathbb E(\u))\,d\x-\max_{\mathbf R\in \mathcal S_{\mathcal L, E}} \mathcal L(\mathbf R\u)\quad &\hbox{if} \ \u\in \mathcal A
\vspace{0.1cm}
\\
 \ \!\!+\infty\quad &\hbox{otherwise in} \ H^1(\om,\mathbb R^3),
\end{array}\right.
\]
where $\mathcal S_{\mathcal L,E} $ is defined by \eqref{kernel}.
 
\begin{remark}\lab{cont} \rm It is worth noticing that \,$\widetilde{\mathcal G}^I\le \mathcal G^I\le \mathcal E^I$, since $\mathbf I\in\mathcal S_{\mathcal L,E}$, and it is straightforward to check that $\mathcal I^I,\ \widetilde{\mathcal G}^I,\ \mathcal G^I$ are all continuous with respect to the strong convergence in $H^1(\om;\mathbb R^3)$.
\end{remark}
%
\subsection{The variational convergence result}
The main result is stated in the next theorem, referring  to \eqref{qmin} for the notion of minimizing sequence.  \MMM
The technical assumption that $\partial\om$ has a finite number of connected components will be needed for applying an extension theorem about divergence-free vector fields from \cite{KMPT}.\KKK
\begin{theorem}
\label{mainth1comp}Assume that $\partial\Omega$ has a finite number of connected components, that
 \eqref{H2 on E}, \eqref{framind}, \eqref{Z1}, \eqref{reg}, \eqref{coerc}, \eqref{winc}, \eqref{shear2}, \eqref{L1} hold true and that
$\mathcal L(\mathbf e_3)<0$. 	\MMM Let $(h_j)_{j\in\mathbb N}\subset(0,1)$ be a vanishing sequence.
Let $(\overline{\mathbf y}_j)_{j\in\mathbb N}\!\subset\! H^1(\om,\mathbb R^3)$ be a minimizing sequence
of $\mathcal G_j^I$. If $\mathbf R_j\!\in\!\mathcal M(\overline{\mathbf y}_j)$ for every $j\in\mathbb N$,  then there are $\overline{\mathbf c}_j\in \mathbb R^3$, $j\in\mathbb N$,  such that the sequence \KKK
 \begin{equation*}
\overline{\u}_j(\x)\,:=\,{h_j}^{-1}\mathbf R_j^T\!\left\{
\big(\overline{\mathbf y}_j-\overline{\mathbf c}_j-\mathbf R_j\x\big)_\alpha\,\mathbf e_\alpha\,+\,(\overline{ y}_{j,3}-x_3)\mathbf e_3\,\right\}
\end{equation*}
is weakly compact in $H^1(\om,\mathbb R^3)$. Therefore up to subsequences, $\overline{\uu}_j\wconv \overline\u$ 
in $H^1(\om,\mathbb R^3)$
and also
\begin{equation*}   \mathcal G_j^I(\mathbf y_{j})\to \widetilde{\mathcal G}^I(\overline\u)= \min_{H^1(\om,\R^3)}\widetilde{\mathcal G}^I =\min_{H^1(\om,\R^3)}\mathcal G^I,\quad\qquad\mbox{as}\ j\to+\infty. 
\end{equation*}
\end{theorem} 
\begin{remark}\rm Since $\widetilde{\mathcal G}^I\le \mathcal G^I$ then equality  $\min\widetilde{\mathcal G}^I =\min\mathcal G^I$ is equivalent to $\argmin \mathcal G^I\subset \argmin \widetilde{\mathcal G}^I$ with possible strict inclusion. 
\end{remark} 
%
%

\section{\textbf{Proof of the variational convergence result}}\label{proofssection}
This section contains the proof of our main result.
We start by showing that sequences of deformations with equibounded energy correspond (up to suitably tuned rotations and translations of the horizontal components) to 
displacements that are equibounded in $H^1$.
\begin{lemma} (\textbf{compactness})\label{lemma compactness}
Assume that $E$, $\mathcal L$ and $\mathcal W$
fulfil \eqref{H2 on E}, \eqref{framind}, \eqref{Z1}, \eqref{reg}, \eqref{coerc}, \eqref{winc}, \eqref{L1} and $\mathcal L(\mathbf e_3) < 0$.
Let $(h_j)_{j\in\mathbb N}\subset(0,1)$ be a vanishing sequence.  \MMM Let $(\yy_j)_{j\in\mathbb N}\subset H^1(\om;\mathbb R^3)$ satisfy $\sup_{j\in\mathbb N}\mathcal G_j^I(\yy_j)<+\infty$. If
 $ {\mathbf R}_j\in \mathcal M(\yy_j)$ for every $j\in\mathbb N$, then
 every limit point of the sequence $(\mathbf R_j)_{j\in\mathbb N}$ belongs to 
 $\mathcal S_{\mathcal L, E},{\color{blue}}$ \KKK
 and by setting 
\begin{equation}\lab{calfa}
 c_{j,\alpha} = \,
 |\Omega|^{-1}\!\!\int_\Omega (\yy_j (\xx)-\mathbf R_j \xx)_\alpha\,d\xx,
\qquad
\alpha=1,2,
\end{equation}
\begin{equation}\displaystyle\lab{c3}
 c_{j,3} =
 -\min_{\xx\in E_{ess}} (\mathbf R_j\xx)_3,
\end{equation}
the sequence
$$ {h_j}^{-1}\big(\,\yy_j\,-\,\mathbf R_j\xx\,-\mathbf c_j\,\big)_{\alpha}\mathbf e_{\alpha}+{h_j}^{-1}(y_{j,3}-x_3)\mathbf e_3$$
is bounded in $H^1(\om;\mathbb R^3)$.  Moreover, $\displaystyle\inf_{j\in\mathbb N}\inf_{H^1(\om;\R^3)}\mathcal G^I_j > -\infty$.
\end{lemma}
\begin{proof} \MMM Since $\mathcal G_j^I\ge \mathcal G_j$ the proof follows from the analogous  \cite[Lemma 4.1]{MPTOB} and we report here that proof for the sake of completeness and also making some steps more precise. \KKK

\textbf{Step 1}. \MMM We start by providing some estimates using the finiteness of $M:=\sup_{j\in\mathbb N}\mathcal G^I_j(\mathbf y_j)$. \KKK
Referring to \eqref{roty}, let  $ {\mathbf R}_j\in \mathcal M(\yy_j)$ for every $j\in\mathbb N$.
 Up to subsequences, 
${\mathbf R}_j\to\mathbf R$ for some $\mathbf R\in SO(3)$. Then we define $\mathbf c_j=(c_{j,1},c_{j,2},c_{j,3})$ by \eqref{calfa} and \eqref{c3}.
By
the rigidity inequality \eqref{muller} there exists a constant $C=C(\om) >0$ such that 
\begin{equation}\begin{array}{ll}
\label{estimate 0}
\displaystyle M&\displaystyle \ge\ \mathcal G_j^I(\yy_j) \,\ge\, 
 C\,  h_j^{-2}\!\! \int_\Omega |\nabla \yy_j- \mathbf R_j|^2 \,d\xx- h_j^{-1} 
 \mathcal L(\yy_j-\xx)\ 
\vspace{0.1cm} \\
&\displaystyle=\, C\,  h_j^{-2} \!\!\int_\Omega |\nabla \yy_j- \mathbf R_j|^2 \,d\xx- h_j^{-1} 
 \mathcal L(\yy_j-\mathbf R_j\xx-\mathbf c_j)
 -
 h_j^{-1} 
 \mathcal L(\mathbf R_j\xx-\xx+\mathbf c_j).
 \end{array}
\end{equation}
Thus, by \eqref{L1} and the definition of $c_{j,3}$ we get
\begin{equation}
\label{estimate 1}
 M \,\ge\, C\,  h_j^{-2} \int_\Omega |\nabla \yy_j- \mathbf R_j|^2 \,d\xx- h_j^{-1} 
 \mathcal L(\yy_j-\mathbf R_j\xx-\mathbf c_j)\,
\end{equation}
and Poincar\'e inequality (with constant denoted by $C_P$) along with Young inequality entail, for every $\varepsilon\!>\!0$, 
\begin{equation}\begin{aligned}
 \label{estimate 2}
 \displaystyle 
 h_j^{-1} \!
 \mathcal L\big(\,(\yy_j-{\mathbf R}_j\xx-{\mathbf c}_j)_\alpha \,{\mathbf e}_{\alpha} \big)
 & \le\, h_j^{-1} C_P \|\mathcal L\|_{*}
\Big ( \sum_{\alpha=1}^{2}\! \int_\Omega |\,(\nabla{\yy_j}-{\mathbf R}_j)_\alpha\, |^2\,d\xx \Big )^{1/2}
 \\
 &\displaystyle \le
 \frac {C_P \|\mathcal L\|^2_{*}} {2\varepsilon}
 \,+\,
 \frac {\varepsilon\,h_j^{-2}\,C_P} {2}\,
 \sum_{\alpha=1}^{2}\int_\Omega |(\nabla\yy_j-{\mathbf R}_j)_{\alpha}|^2\,d\xx. 
 \\
 \end{aligned}
\end{equation}
Estimates \eqref{estimate 1} and \eqref{estimate 2} together with Young inequality provide
\begin{eqnarray*}
\quad M \!\!\!\!&\ge&\!\!\!\! h_j^{-2}\! \left(\!C-\frac {\varepsilon\, C_P}{2}\!\right) \!\! \int_\Omega \!|\nabla \yy_j\!-\! \mathbf R_j|^2 d\xx
 -\frac {C_P \,\|\mathcal L\|^2_{*}}{2\,\varepsilon}
 - h_j^{-1} \!\mathcal L\big((\yy_j-\mathbf R_j\xx-\mathbf c_j)_3 \,\mathbf e_3 \big)\! 
 \\
 \nonumber
 \!\!\!\!&\ge&\!\!\!\! h_j^{-2}\! \left(\!C-\frac {\varepsilon\, C_P}{2}\right)  \!\int_\Omega \!|\nabla \yy_j\!-\! \mathbf R_j|^2 \,d\xx
 -\frac {C_P \|\mathcal L\|^2_{*}} {2\,\varepsilon}
 \\
 \nonumber
 && \ \
 - h_j^{-1} \,\|\mathcal L\|_{*} \,
 \big (\,
\| (\yy_j-\mathbf R_j\xx-\mathbf c_j)_3 \|_{L^2(\Omega)} 
+
\| \nabla (\yy_j-\mathbf R_j\xx)_3  \|_{L^2(\Omega)} 
 \,\big)
   \\
 \nonumber
\!\!\!\!&\ge&\!\!\!\! h_j^{-2}\! \left(\!C-\frac {\varepsilon\,C_P}{2}-\frac {\varepsilon} 2\right)  \!\int_\Omega \!|\nabla \yy_j\!-\! \mathbf R_j|^2 \,d\xx
 -\left(\frac {C_P } {2\,\varepsilon} + \frac 1 {2\,\varepsilon}\right) \|\mathcal L\|^2_{*} 
 \\
 \nonumber
 && \ \
 - h_j^{-1} \,\|\mathcal L\|_{*} \,
 \big (\,
\| (\yy_j-\mathbf R_h\xx-\mathbf c_j)_3 \|_{L^2(\Omega)}  .
 \end{eqnarray*}
 By choosing $\varepsilon= C/ (C_P+1)$, we get
 \beeq\begin{aligned}
 \label{estimate 4}
 &h_j^{-2}\ \frac {C}{2}  \int_\Omega \!|\nabla \yy_j\!-\! \mathbf R_j|^2 \,d\xx \\&\qquad
 \le\, M\, +\,
 \frac {(C_P+1)^2}{2C}
 \|\mathcal L\|^2_{*} 
 \,+\,
 h_j^{-1} \,\|\mathcal L\|_{*} \,
\| (\yy_j-\mathbf R_j\xx-\mathbf c_j)_3 \|_{L^2(\Omega;\mathbb R^3)}.
\end{aligned}\eneq
Thus, if we show that $\ h_j^{-1} 
\| (\yy_j-\mathbf R_j\xx-\mathbf c_j)_3 \|_{L^2(\Omega)}\,$ is 
uniformly bounded, then, due to estimate \eqref{estimate 4},
$\|h_j^{-1}(\nabla\yy_j-\mathbf R_j)\|_{L^2(\Omega)}$ is uniformly bounded too and Poincar\`e inequality entails uniform boundedness of $ h_j^{-1}
\| \yy_j-\mathbf R_j\xx-\mathbf c_j \|_{H^1(\Omega;\mathbb R^3)}$.

Let us define for every $j$
\beeq\label{tw}
 t_j\ :=\ h_j^{-1} \,
\| \,(\yy_j-\mathbf R_j\xx-\mathbf c_j)_3\, \|_{L^2(\Omega)} \qquad\mbox{and}\qquad  w_j \ :=t_j^{-1} h_j^{-1} \,
 (\yy_j-\mathbf R_j\xx-\mathbf c_j)_3.
\eneq
Then
\begin{equation}
 \label{estimate 4bis}
 \|w_j\|_{L^2(\Omega)}=1,\qquad |\nabla \yy_j\!-\! \mathbf R_j|^2\ =\
 \sum_{\alpha=1}^2 \,|\,\nabla (\yy_j\!-\! \mathbf R_j\xx)_\alpha\,|^2
 \ +\ 
h_j^{\,2} t_j^{\,2}\,|\nabla w_j|^2,
\end{equation} so that \eqref{estimate 0}, \eqref{estimate 2} and \eqref{L1} imply that for every $\eps>0$
\beeq
\begin{aligned}
 \label{estimate 5}
 & C t_j^{\,2}\int_\Omega|\nabla w_j|^2\,d\xx \,-\, t_j\,\mathcal L(w_j\,\mathbf e_3)-h_j^{-1}\mathcal L(\mathbf R_j\x-\x+c_{j,3}\mathbf e_3)\\
&\qquad
\le  M-C\,h_j^{-2}\, \sum_{\alpha=1}^2 \,\int_\Omega|\,
\nabla (\yy_j\!-\! \mathbf R_j\xx)_\alpha\,|^2\,d\xx\,+
\,h_j^{-1} \,
\mathcal L \big( \,(\yy_j\!-\! \mathbf R_j\xx-\mathbf c_j)_\alpha\,\mathbf e_\alpha\,\big)\\
&\qquad\le M-C\,h_j^{-2}\, \sum_{\alpha=1}^2 \,\int_\Omega|\,
\nabla (\yy_h\!-\! \mathbf R_j\xx)_\alpha\,|^2\,d\xx\,+
\, \frac {C_P\,\|\mathcal L\|^2_{*} }{2\,\varepsilon}\\&\qquad\quad
\,+\,
\frac {h_j^{-2}\, \varepsilon\,C_P} {2} \sum_{\alpha=1}^2\int_\Omega |\nabla (\yy_j\!-\! \mathbf R_j\xx)_\alpha\,|^2\,d\xx, 
\end{aligned}\eneq
and by choosing $\varepsilon=2C/C_P $ in \eqref{estimate 5}
we get
\begin{equation}
  \label{estimate 6}
  C\,
  t_j^{\,2}\int_\Omega|\nabla w_j|^2\,d\xx \,-\, t_j\,\mathcal L(w_j\,\mathbf e_3)-h_j^{-1}\mathcal L(\mathbf R_j\x-\x+c_{j,3}\mathbf e_3)\ \le\
  \frac{C_P^2\,\|\mathcal L\|^2_{*}}{4\,C}+M
\end{equation}
while, by choosing $\varepsilon=C/C_P ,$  \eqref{estimate 5} yields
\begin{eqnarray*}
  && 
  \frac 1 2 \,C\,h_j^{-2}\, \sum_{\alpha=1}^2 \,\int_\Omega\,
\nabla (\yy_j\!-\! \mathbf R_j\xx)_\alpha\,|^2\,d\xx\,+\,C\,
t_j^{\,2}\,\int_\Omega\,|\nabla w_j|^2\,-t_j\,
\mathcal L(w_j\,\mathbf e_3)\ \le
\\
\nonumber
&& \qquad \le \
\frac {C_P^2}{2C}\,\|\mathcal L\|^2_{*}+M +h_j^{-1}\mathcal L(\mathbf R_j\x-\x+c_{j,3}\mathbf e_3).
\end{eqnarray*}
Thus, by taking account of the fact that \eqref{L1} entails  $\mathcal L(\mathbf R_j\x-\x+c_{j,3}\mathbf e_3)\le 0$,  we get
\begin{equation}
  \label{estimate 8}
 \frac 1 2 \,C\,\frac {h_j^{-2}}{t_j^{\,2}}\, \sum_{\alpha=1}^2 \,\int_\Omega|\,
\nabla (\yy_j\!-\! \mathbf R_j\xx)_\alpha\,|^2\,d\xx\ \le\
\frac 1 {t_j}\,
\mathcal L(w_j\,\mathbf e_3)
\,+\,
\frac 1 {{t_j}^{\,2}}\,\frac {C_P^2}{2C}\,\|\mathcal L\|^2_{*}+\frac{M}{t_j^2}.
\end{equation}
Moreover, \eqref{estimate 6} and \eqref{L1} yield
 \begin{equation}
  \label{estimate 17}
 - h_j \frac{C_P^2\,\|\mathcal L\|^2_{*}}{4\,C}-Mh_j
-\, h_jt_j\,\mathcal L(w_j\,\mathbf e_3)\le \Phi(\mathbf R_j, E, \mathcal L)= \mathcal L(\mathbf R_j\x-\x+c_{j,3}\mathbf e_3)\ \le 0.
 \end{equation}

\MMM
{\textbf {Step 2}}.
 Here we prove that $t_j$ from \eqref{tw} is 
uniformly bounded, thus yielding as shown through the previous step the uniform boundedness of $ h_j^{-1}
\| \yy_j-\mathbf R_j\xx-\mathbf c_j \|_{H^1(\Omega;\mathbb R^3)}$.
 In order to check that the sequence $(t_j)_{j\in\mathbb N}$ is 
uniformly bounded
 we assume by contradiction that, up to subsequences, $\lim_{j\to+\infty}t_j=+\infty$.  \KKK
Normalization $\|w_j\|_{L^2(\om)}=1$ entails, for every $\varepsilon>0$,
\begin{eqnarray*}
   \mathcal L \big( w_j\,\mathbf e_3\big)
   &\le& \|\mathcal L\|_{*} \big(
   \,\|w_j\|_{L^2}+ \| \nabla w_j \|_{L^2}\,\big)
   \,=\,
   \|\mathcal L\|_{*} \big(
   \,1+ \| \nabla w_j \|_{L^2}\,\big)
   \\
   \nonumber &\le&
    \|\mathcal L\|_{*}\,\,+\,
    \frac {\|\mathcal L\|^2_{*}}{2\,\varepsilon}\,+\,
  \frac{\varepsilon}2
    \, \| \nabla w_j \|_{L^2}^2,
\end{eqnarray*}
and choosing $\varepsilon=C\,t_j^2$ we get, by \eqref{estimate 6},  
\begin{equation}
\label{estimate 10}
\frac {C}2t_j^2\,\int_\Omega|\nabla w_j|^2\,\,d\xx\ \le \ 
t_j\,\|\mathcal L\|_{*}\,\,+\, \frac {\|\mathcal L\|^2_{*}} {2C\,t_j}+M,
\end{equation}
thus $\int_\Omega |\nabla w_j|^2\,\,d\xx\to 0$, so by \eqref{estimate 4bis} $w_j\to w $  in $H^1(\Omega;\R^3)$ with $\nabla w= 0$ a.e. in $\om$,   that is, $w$  is a constant function since $\om $ is a connected open set. 
Combining estimates \eqref{estimate 8} and \eqref{estimate 10} we get
\begin{eqnarray*}
&&
 \frac 1 2  C\,\frac {h_j^{-2}}{t_j^{\,2}}\, \sum_{\alpha=1}^2 \,\int_\Omega|\,
\nabla (\yy_j\!-\! \mathbf R_j\xx)_\alpha\,|^2\,d\xx
\\ \nonumber
&& \qquad \le \
\frac 1 {t_j}
\left(
\|\mathcal L\|^2_{*}\,+\,
    \frac {\|\mathcal L\|^2_{*}}{2}\,+\,
    \frac {1}{2}\, \| \nabla w_j \|_{L^p}^2
\right)
\,+\,
\frac 1 {{t_j}^{\,2}}\,\frac {C_P^2}{2\,C_R\,C}\,\|\mathcal L\|^2_{*}+\frac{M}{t_j^2},
\end{eqnarray*}
hence 
\begin{equation*}
\frac{1}{h_j\,t_j}\,\nabla \big( \yy_j-\mathbf R_j\xx\big)_\alpha\ \to \ 0 \qquad \hbox{in }L^2(\Omega;\R^{3\times 3}) \qquad \hbox{if }\alpha=1,2,
\end{equation*}
and by definition of $w_j$ we have
\begin{equation*}
h_j^{-1}t_j^{-1}\,\left(\yy_j-\mathbf R_j\xx-\mathbf c_j\right)_3\ \to\ w \qquad \hbox{q.e. } \xx\in E\,.
\end{equation*}
Moreover, by \eqref{estimate 6} and \eqref{L1} we get
\begin{equation*}
\mathcal L ( w_j\,\mathbf e_3) \ \ge \ -\,\frac {C_P^{\,2}}{4\,C\, t_j}\, \|\mathcal L\|^2_{*}-\frac{M}{t_j}.
\end{equation*}
Hence, due to $\mathcal L ( w_j\,\mathbf e_3)\to
\mathcal L ( w\,\mathbf e_3)=w\,\mathcal L(\mathbf e_3)$, we have $w\,\mathcal L (\mathbf e_3)\!\ge\! 0,$ thus, by taking  account of $\mathcal L (\mathbf e_3) < 0,$ we get $w\le 0$ and eventually, by $\|w_j\|_{L^2}=1$, we obtain $w\ < \ 0$.

 We notice now that if we set $L:=\liminf_{j\to+\infty} h_jt_j$, then  either $ L\in (0, +\infty]$ or $L= 0$. Assume first that $ L\in (0, +\infty)$. Since on a subsequence $\mathbf R_j\to \mathbf R\in SO(3)$,  we have
for q.e. $\xx\!\in\! E_{ess}$
\begin{equation}\begin{array}{rl} \label{estimate 16a}
\displaystyle\frac {y_{j,3}^*} {\ h_j\,t_j}=
&
\displaystyle  \frac {y_{j,3}^*-(\mathbf R_j\xx)_3-\mathbf c_{j,3}} {\ h_j\,t_j}
 \,+\,
 \frac {(\mathbf R_j\xx)_3+\mathbf c_{j,3}} {\ h_j\,t_j}\to
 \displaystyle w + L^{-1}\{(\mathbf R\xx)_3 -\min_{E_{ess}} (\mathbf R\xx)_3 \}
 \end{array} \end{equation}
 \MMM as $j\to+\infty$ (along a suitable subsequence). \KKK 
Since $$\displaystyle \min_{E_{ess}} \left\{(\mathbf R\xx)_3 -\min_{E_{ess}} ((\mathbf R\xx)_3 )\right\}=0$$ then there exists $E'\subset E_{ess}$ with $\capacity E' > 0$ such that $$(\mathbf R\xx)_3 -\min_{E_{ess}} (\mathbf R_xx)_3  < -\frac{wL}{2}$$
on $E'$, so that by \eqref{estimate 16a}
\[ \displaystyle\frac {y_{j,3}^*} {\ h_j\,t_j}\to w + L^{-1}\{(\mathbf R\xx)_3 -\min_{E_{ess}} (\mathbf R\xx)_3 \} < \frac{w}{2}< 0 \]
 for q.e. $\x\in E',$ a contradiction since
\MMM the assumption $\sup_{j\in\mathbb N}\mathcal G_j^I(\mathbf y_j)<+\infty$ implies $\mathbf y_j\in\mathcal A$, i.e., $y_{j,3}^*\ge 0$ q.e. in $E$. \KKK
  If $L= +\infty$ then  by arguing as in estimate \eqref{estimate 16a} we easily get $h_j^{-1}t_j ^{-1}y_{j,3}^*\to w < 0$ for q.e. $\x\in E'$ which is again a contradiction. Therefore we are left to assume that, up to subsequences, $h_jt_j\to 0$. 

In order to complete this step
we notice that either $\mathbf R_j\mathbf e_3\not=\mathbf e_3$ for $j$ large enough or $\mathbf R_j\mathbf e_3=\mathbf e_3$ for infinitely many $j$, and we separately treat these two cases. In the first case, by taking  account that
$\mathcal L(\mathbf e_3)<0,$
 \cite[Lemma 3.10]{MPTOB} entails 
\begin{equation}
 \label{estimate 15}
 \limsup_{j\to +\infty} \
\frac
{\displaystyle
\Phi(\mathbf R_j,E,\mathcal L)}{|\mathbf R_j\mathbf e_3-\mathbf e_3 |}
\ < \ 0.
\end{equation}
By \eqref{estimate 17} 
we get therefore
\begin{equation}\lab{gamma}
\gamma:=\liminf _{j\to +\infty}\frac {h_jt_j}{|\mathbf R_j\mathbf e_3-\mathbf e_3 |} > 0
\end{equation}
and for large enough $j$ \eqref{gamma} yields
\begin{equation*}
 |\mathbf R_j\mathbf e_3-\mathbf e_3 |\,\le\, \, \frac{2h_jt_j}{\gamma} \,.
\end{equation*} 
Hence, if $\displaystyle\x_j\in \argmin_{E_{ess}} \{ (\mathbf R_j\xx)_3 \}$ for every $j$, so we may assume that, up to subsequences, $\x_j\to \overline\x\in E_{ess}$, then  for every $\x\in E_{ess}$
\[\begin{aligned} \displaystyle\theta_j(\x):&=(\mathbf R_j\xx)_3-\min_{\x\in E_{ess}} ((\mathbf R_j\xx)_3 )=( \mathbf R_j(\xx-\x_j))_3=( \mathbf R_j(\xx-\x_j))_3-(\x-\x_j)_3\\
&\displaystyle=(\x-\x_j)\cdot(\mathbf R_j^T-\mathbf I)\mathbf e_3\le |\x-\x_j||\mathbf R_j^T\mathbf e_3-\mathbf e_3|\le  \frac{2h_jt_j}{\gamma}|\x-\x_j|
\end{aligned}
\]
and then we get, for $j$ large enough,  
\beeq\label{thetaj} \theta_j(\x)\le \frac{2h_jt_j}{\gamma}|\x-\overline\x|\le \frac{h_jt_j}{4}|w|
\qquad\mbox{for every $\x\in E'':= B_{\frac{\gamma |w|}{16}}(\overline\x)\cap E_{ess}$.}\eneq
Moreover, by taking  account  of $w_j\to w$ in $H^1(\om;\mathbb R^3)$ we get  
\begin{equation*}
\displaystyle \yy_{j,3}^*=
  h_jt_jw_j
 \,+ \theta_j
  = h_jt_jw+\theta_j+\delta_j\qquad\mbox{where $\delta_j:=h_jt_j(w_j-w)$},
 \end{equation*}
 with $h_j^{-1}t_j^{-1}\delta_j\to 0$ in $H^1(\om;\mathbb R^3)$ and thus, up to subsequences, $\capacity$ quasi uniformly in $\om$. Since $\capacity E'' > 0$ we may assume that there exists $E'''\subset E''$ with $\capacity E''' > 0$ such that $h_j^{-1}t_j^{-1}\delta_j\to 0$ uniformly in $E'''$ hence for $j$ large enough
 $$w+h_j^{-1}t_j^{-1}\delta_j(\x) < \frac{w}{2}\quad \hbox{q.e.}\ \x\in E'''.$$ 
 From the latter estimate and \eqref{thetaj} we get for $j$ large enough
 \[
 \displaystyle h_j^{-1}y_{j,3}^*(\x)\le \frac{t_j}{4}w  < \frac w4<0 \qquad \hbox{q.e.}\ \x\in E''',
 \]
 a contradiction since 
$y_{j,3}^*(\x)\ge 0$ for q.e. $\xx\in E$. \MMM
 In the second case we may assume that $\mathbf R_j\mathbf e_3=\mathbf e_3$ for every $j$ so $c_{j,3}=0$ and $x_3=0\Rightarrow (\mathbf R_j\x)_3=0$ for every $j$. In particular $w_j\to w<0$ in $H^1(\Omega;\R^3)$ implies
 \[
 t_j^{-1}h_j^{-1}y_{j,3}(\x)=t_j^{-1}h_j^{-1}(y_{j,3}(\x)-(\mathbf R_j \x)_3-c_{j,3})=w_j(\x)\to w<0\qquad \mbox{q.e.  $\x\in E$}.
 \] 
again contradicting $y_{j,3}^*(\x)\ge 0$ for q.e. $\xx\in E$.
This proves that $(t_j)_{j\in\mathbb N}$ is a bounded sequence.\\

{\bf{Step 3.}} We have shown that $ h_j^{-1}
\| \yy_j-\mathbf R_j\xx-\mathbf c_j \|_{H^1(\Omega;\mathbb R^3)}$ is uniformly bounded, yielding in particular that $(t_jw_j)_{j\in \mathbb N}$ is a bounded sequence in $H^1(\Omega;\mathbb R^3)$. But then \eqref{estimate 17} yields
\begin{equation*}
 -M'h_j\le \Phi(\mathbf R_j, E, \mathcal L)= \mathcal L(\mathbf R_j\x-\x+c_{j,3}\mathbf e_3)\ \le 0
\end{equation*}
for some suitable constant $M'>0$ independent of $j$. If, up to subsequences, $\mathbf R_j\to \mathbf R$, we get
$
\Phi(\mathbf R_j, E, \mathcal L)\to \Phi(\mathbf R, E, \mathcal L)=0
$
thus proving that
$\mathbf R\in \mathcal S_{\mathcal L, E}$. \KKK

Therefore, in order to end the proof, we have  to show   that  the sequence $(h_j^{-1}(y_{j,3}-x_3))_{j\in\mathbb N}$ is equibounded in $H^1(\Omega;\R^3)$. \MMM We claim that the sequence $(h_j^{-1}|\mathbf R_j\mathbf e_3-\mathbf e_3|)_{j\in\mathbb N}$ is bounded. Indeed,
 if we assume that $\mathbf R_j\mathbf e_3\not = \mathbf e_3$ for every $j$ large enough,  \eqref{estimate 15}  entails
 \[
\gamma':= \liminf_{j\to+\infty}\frac{h_j}{|\mathbf R_j\mathbf e_3-\mathbf e_3|}>0, 
 \]
 so that
 $|\mathbf R_{j}\mathbf e_3- \mathbf e_3|\le {2h_j }/{\gamma'}$ for $j$ large enough, thus proving the claim. \KKK
 As a consequence, since
\[\begin{aligned}
\displaystyle\left |h_{j}^{-1}(y_{j,3}-x_3)\right |
\displaystyle&\le \left |h_{j}^{-1}(y_{j,3}-(\mathbf R_{j}\xx)_3- c_{j,3})\right |+h_{j}^{-1} |(\mathbf R_{j}\xx)_3-x_3|
\\&\le \left |h_{j}^{-1}(y_{j,3}-(\mathbf R_{j}\xx)_3- c_{j,3})\right |+h_{j}^{-1}|\mathbf R_{j}\mathbf e_3- \mathbf e_3|\,|\xx|
\end{aligned}
\]
and
\[\begin{aligned}
\displaystyle\left |h_{j}^{-1}\nabla(y_{j,3}-x_3)\right |&\le
\displaystyle h_{j}^{-1}\left |\nabla(y_{j,3}-(\mathbf R_{j}\x)_3)\right|+h_j^{-1}|\nabla((\mathbf R_j\x)_3-x_3)|\\&\le h_j^{-1}\left |\nabla(y_{j,3}-(\mathbf R_{j}\x)_3)\right|+ h_{j}^{-1} |\mathbf R_{j}\mathbf e_3-\mathbf e_3|,
\end{aligned}
\]
we obtain that the sequence $(h_j^{-1}(y_{j,3}-x_3))_{j\in\mathbb N}$ is indeed equibounded in $H^1(\Omega;\R^3)$.\\

\MMM
{\bf{Step 4}}.
Eventually, we prove the last statement by contradiction, assuming that \KKK there exists a sequence $(\overline\yy_j)_{j\in\mathbb N}\subset {\mathcal A}$ such that $\mathcal G_j^I(\overline\yy_j)\to -\infty$. In such case,  $\mathcal G_j^I(\overline\yy_j)\le 0$ for $j$ large enough so, as we have just proved, there exist sequences $(\mathbf R_j)_{j\in\mathbb N}\subset SO(3)$ and $(\mathbf {\bar c}_j)_{j\in\mathbb N}\subset \mathbb R^3$ and there exists $\v\in H^1(\om;\R^3)$ such that, up to subsequences, 
$$\v_j:=h_j^{-1}(\overline\yy_j-\mathbf R_j\xx-\mathbf {\bar c}_j)\wconv \v$$
weakly in $H^1(\om ;\mathbb R^3)$. Therefore by \eqref{L1}
$$\displaystyle\liminf_{j\to +\infty}\mathcal G_j^I(\overline\yy_j)\ge -\limsup_{j\to +\infty}\mathcal L(\v_j)=-\mathcal L(\v)> -\infty,$$
a contradiction.
\end{proof}

\begin{lemma}
\label{Lemma 3.8} (\textbf{Lower bound})
Assume that $E$, $\mathcal L$, $\mathcal W$
fulfill the conditions \eqref{H2 on E}, \eqref{framind}, \eqref{Z1}, \eqref{reg}, \eqref{coerc}, \eqref{winc}, \eqref{shear2}, \eqref{L1} and $\mathcal L(\mathbf e_3) < 0$.
Let $(h_j)_{j\in\mathbb N}\subset(0,1)$ be a vanishing sequence. Let $(\yy_j)_{j\in\mathbb N}\subset H^1(\om;\mathbb R^3)$ satisfy  $\sup_{j\in\mathbb N}\mathcal G_j^I(\yy_j)<+\infty$. Let   $\mathbf R_j\in \mathcal M(\yy_j)$ for every $j\in\mathbb N$. 
  If
\beeq\lab{seq}\u_j(\x)\,:=\,{h_j}^{-1}\mathbf R_j^T\!\left\{
\big(\mathbf y_j-\mathbf c_j-\mathbf R_j\x\big)_\alpha\,\mathbf e_\alpha\,+\,(y_{j,3}-x_3)\mathbf e_3\,\right\},\quad j\in\mathbb N,
\eneq
where $\mathbf c_j$ are defined by \eqref{calfa}-\eqref{c3},
then,  up to passing to a not relabeled subsequence, we have $\u_j\wconv \u$ weakly in $H^1(\Omega;\R^3)$, $\u\in\mathcal A$ and
\begin{equation*}
\displaystyle\liminf_{j\to +\infty}\mathcal G_j^I(\yy_j)\ge \widetilde{\mathcal G}^I(\u).
\end{equation*}
\end{lemma}
\begin{proof} Finiteness of $\mathcal G^I_j(\yy_j)$ imply $\yy_j\in\mathcal A$ for every $j\in\mathbb N$.
Due to Lemma \ref{lemma compactness}, the sequence defined in \eqref{seq} is equibounded in $H^1(\om;\mathbb R^3)$ hence there exists $\u\in H^1(\om;\mathbb R^3)$ such that up to subsequences $\u_j\wconv\u$ in $H^1(\om;\mathbb R^3)$.\\
Moreover since
\begin{equation*}\begin{aligned}
1&=\det\nabla\mathbf y_j=\det(\mathbf R_j(\mathbf I+h_{j}\nabla\u_{j}))=\det(\mathbf I+h_{j}\nabla\u_{j})=\\
&= 1+h_{j}\dv\u_j-\frac{1}{2}h_{j}^{2}(\mathrm{Tr}(\nabla\u_{j})^{2}-(\Tr \nabla\u_{j})^{2})+h_{j}^{3}\det \nabla\u_j
\end{aligned}
\end{equation*}
a.e. in $\om$ we get
\begin{equation*}\lab{TrB}
\dv \u_{j}=
\frac{1}{2}h_{j}(\mathrm{Tr}(\nabla\u_{j})^{2}-(\Tr \nabla\u_{j})^{2})-h_{j}^{2}\det \nabla\u_j.
\end{equation*}
By taking into account that $\nabla\u_j$ are uniformly bounded in $L^2$ we get $h_j^{\alpha}|\nabla\u_j|\to 0$ a.e. in $\om$ for every $\alpha > 0$ hence
 $\dv \u_{j}=\frac{1}{2}h_{j}(\mathrm{Tr}(\nabla\u_{j})^{2}-(\Tr \nabla\u_{j})^{2})-h_{j}^{2}\det \nabla\u_j\to 0$ a.e. in $\om$. 
Since the weak convergence of $\nabla\u_j$ implies $\dv \u_{j}\wconv \dv \u$ weakly in $L^2(\Omega)$ we get $\dv\u=0$ a.e. in $\om$. 
By recalling \cite[Lemma A1]{BBGT}
we get, again up to subsequences, $\u_j^*(\xx)\to \uu^*(\xx)$ for q.e. $\xx\in E$ hence by taking  account of  
$$ u_{j,3}^*(\x)={h_j}^{-1} y_{j,3}^*(\x)\ge 0$$
 for q.e. $\xx\in E$ we get $ u_{3}^*(\x)\ge 0$   for q.e. $\xx\in E$ that is $\uu\in \mathcal A$.
 By taking account of  $\mathcal G_j^I\ge \mathcal G_j$, of $\dv\uu=0$, and of $\mathbf R_j\to\mathbf R\in \mathcal S_{\mathcal L,E}$ up to subsequences as shown in Lemma \ref{lemma compactness}, we may invoke  \cite[Lemma 4.3]{MPTOB}, which  entails 
 that
\begin{equation*}\begin{aligned}
\displaystyle\liminf_{j\to +\infty}\mathcal G_j^I(\yy_j)&\ge\liminf_{j\to +\infty}\mathcal G_j(\yy_j)\\
& \displaystyle\ge \min_{\mathbf b\in \mathbb R^2}\int_\om \mathcal Q(\x, \mathbb E(\u)+\tfrac{1}{2}b_{\alpha}(\mathbf e_{\alpha}\otimes\mathbf e_3+\mathbf e_{3}\otimes\mathbf e_\alpha))\,d\x-\mathcal L(\mathbf R\u)\ge \widetilde{\mathcal G}^I(\u)
\end{aligned}
\end{equation*}
thus proving the claim.
\end{proof}
\begin{remark}\lab{shear3} \rm If condition \eqref{shear2} is not satisfied then the claim of Lemma \ref{Lemma 3.8} may fail. Indeed, the example \cite[Remark 4.5]{MPTOB} applies also to the incompressible case.
\end{remark}

  \MMM
In the next auxiliary lemma we provide some estimates for  the Lagrangian flow 
\beeq \lab{flowbis}\left\{\begin{array}{ll} &\displaystyle\frac{\partial  {\bf z}}{\partial t}(t,\x)=\v({\bf z}(t,\x))\qquad t> 0\\
&\\
& {\bf z}(0,\x)=\x
\end{array}\right.  
\eneq
 associated to a vector field $\v\in C_c^1(\R^3,\R^3)$. 
  
\begin{lemma}\label{newflowlemma}
 Let $\v\in C_c^1(\R^3,\R^3)$ be satisfying $\mathrm{div}\, \v=0$  in an open neighborhood $\om_*$ of $\overline\om$.
Let  $\zz\in C^1([0,+\infty)\times \R^3; \mathbb R^3)$ be the unique global solution to \eqref{flowbis}. Then there exists $T>0$ such that
 \begin{equation}
\zz(t,\x)\in\Omega_*\quad\mbox{and}\quad\det\nabla \zz(t,\x)=1\label{flux1}\quad \quad\forall\ t\in [0,T]\;\;\;\forall \x\in \overline\Omega.
\end{equation}
Moreover \KKK
\begin{equation}\label{nuova1}\MMM
\sup_{\x\in\overline\om}|\zz(t,\x)-\x|\le t\|\v\|_\infty\,\exp(t\|\nabla \v\|_\infty)\qquad\forall\ t\ge 0,\KKK
\end{equation}
\begin{equation}
 \sup_{\x\in\overline\Omega}|t^{-1}(\zz(t,\x)-\x)- \v(\x) |\le \|\v\|_{L^{\infty}}(\exp(t\|\nabla\v\|_{L^{\infty}})-1)\qquad  \forall\ t>0,\label{flux2}
\end{equation}
\begin{equation}\label{nuova2}\MMM
\sup_{\x\in\overline\om}|\nabla\zz(t,\x)|\le 3\exp(t\|\nabla v\|_\infty)\qquad  \forall\ t\ge 0,\KKK
\end{equation}
\begin{equation}\displaystyle
\sup_{\x\in\overline\Omega}|(\nabla\zz(t,\x)-\mathbf I)|\le 3(\exp(t\|\nabla\v\|_{L^{\infty}})-1)\qquad \forall\ t\ge 0,\label{flux3}
\end{equation}
where $\nabla \zz$ denotes the gradient of $\zz$ with respect to the $\x$ variable.
\end{lemma}\KKK
\begin{proof}
\MMM
Since $\v\in C_c^1(\R^3,\R^3)$, \eqref{flowbis}
has a unique global solution  $\mathbf z\in C^1([0,+\infty)\times \R^3;\mathbb R^3)$. 
By integrating \eqref{flowbis} we have  \KKK
\begin{equation}\label{diffbis}
\left(\mathbf z(t,\x)-\x\right)-t\mathbf v(\x)=\int_0^t(\mathbf v(\mathbf z(s,\x))-\mathbf v(\x))\,ds
\end{equation}
for any $\x\in\overline\Omega$ and for every $t \ge 0$, whence
\[
|\mathbf z(t,\x)-\x|\le  t|\v(\x)|+ \|\nabla\mathbf v\|_{\infty}\int_0^t |\mathbf z(s,\x)-\x|\,ds
\]
and Gronwall lemma entails 
\begin{equation}\label{gronw}
|\mathbf z(t,\x)-\x|\le t|\mathbf v(\x)| \exp(t\|\nabla\mathbf v\|_{\infty})\le t\|\vv\|_{\infty}\exp(t\|\nabla\mathbf v\|_{\infty}).
\end{equation}
Therefore there exists $T > 0$ such that $\zz(t,\x)\in \om_*$ for every $(t,\x)\in [0,T]\times\overline\om$. Moreover, \eqref{flux1} follows by a standard argument, since $\v$ is divergence free in $\Omega_*$, see also \cite[Lemma 4.1]{MaiPe1}.
By exploiting  \eqref{diffbis} and \eqref{gronw} we get
\begin{equation*}\begin{aligned}
&\left| t^{-1}(\mathbf z(t,\x)-\x)-\mathbf v(\x)\right|	\le  \|  \v\|_{\infty}\, \int_{0}^{t}s^{-1}\,\left |{\bf z}(s,\x)-\x\right |\,ds\\&\qquad\qquad\le \|\v\|_{\infty}\int_0^t\exp(s\|\nabla\mathbf v\|_{{\infty}})\,ds\le \|\v\|_{{\infty}}(\exp(t\|\nabla\v\|_{{\infty}})-1)
\end{aligned}\end{equation*}
for any $\x\in\overline\Omega$ and any $t\in(0,T]$ thus proving \eqref{flux2}.
Letting $\nabla$ denote the derivative in the $\x$ variable we see that $\mathbf Z(t,\x):=\nabla \mathbf z(t,\x)$  satisfies
\[\left\{\begin{array}{ll} &\displaystyle\frac{\partial  { \mathbf Z}}{\partial t}(t,\x)=\nabla \v({\bf z}(t,\x))\mathbf Z(t,\x)\qquad t>0\\
&\\
& {\mathbf Z}(0,\x)=\mathbf I,
\end{array}\right.
\]
whence
\begin{equation}\label{nuovaextra}
\mathbf Z(t,\x)-\mathbf I=\int_0^t\nabla \v(\zz(s,\x))\,\mathbf Z(s,\x)\,ds
\end{equation}
for every $\x\in\overline \om$ and every $t\ge 0$, therefore
\[ |\mathbf Z(t,\x)|\le 3+\|\nabla\v\|_{\infty}\int_0^t |\mathbf Z(s,\x)|\,ds\]
and by Gronwall Lemma
\[ |\mathbf Z(t,\x)|\le 3\exp(t\|\nabla\v\|_{{\infty}}). \]
Therefore for every $t\in (0,T]$ and every $\x\in\overline\Omega$ we get
\begin{equation*}\begin{array}{ll}
\displaystyle | \mathbf Z(t,\x)-\mathbf I|&\le\displaystyle\int_0^t|\nabla \v(\mathbf z(s,\x))|\,|\mathbf Z(s,\x)|\,ds
\displaystyle\le \|\nabla\v\|_{{\infty}} \,\int_0^t |\mathbf Z(s,\x)|\,ds\\&\le\displaystyle 3\|\nabla\v\|_{{\infty}} \,\int_0^t\exp(s\|\nabla\v\|_{{\infty}})\,ds
\le \,3(\exp(t\|\nabla\v\|_{{\infty}})-1)
\end{array}\end{equation*}
%
%
thus proving \eqref{flux3}. \end{proof}

  \begin{lemma}{\bf (Upper bound)}\lab{upbd} Assume that $\partial\Omega$ has a finite number of connected components, that
 \eqref{H2 on E}, \eqref{framind}, \eqref{Z1}, \eqref{reg}, \eqref{coerc}, \eqref{winc}, \eqref{shear2}, \eqref{L1} hold true and that
$\mathcal L(\mathbf e_3)<0$.
\MMM Let $p>3$\KKK. Let  $\u\in W^{1,p}(\Omega,\mathbb R^3)\ $ such that $\dv \u=0$ a.e.  in $\om$. Then there exists a sequence
 $({\yy}_j)_{j\in\mathbb N}\subset C^1(\overline\Omega,\mathbb R^3)$ such that
\[ \limsup_{j\to +\infty} {\mathcal G}_j^I({\mathbf y}_j) \le \widetilde{\mathcal G}^I(\u).\]
\end{lemma}

\begin{proof}{\bf Step 1.}
We assume without loss of generality that $\u\in \mathcal A$ and we let
$$\mathbf b^*\in \argmin \left\{\int_\om \mathcal Q(\x, \mathbb E(\u)+\frac{1}{2}b_{\alpha}(\mathbf e_{\alpha}\otimes\mathbf e_3+\mathbf e_{3}\otimes\mathbf e_\alpha))\,d\x: \mathbf b\in \mathbb R^2\right\},$$
\beeq\lab{utilde}\widetilde\u(\x):= \u(\x)+x_3(b^*_1\mathbf e_1+b^*_2\mathbf e_2).\eneq
It is readily seen that $\widetilde\u\in \mathcal A$, that $\dv\widetilde \u=0$ a.e. in $\Omega$ and  that $\mathbb E(\widetilde\u)=\mathbb E(\u)+\frac{1}{2}b_{\alpha}^*(\mathbf e_{\alpha}\otimes\mathbf e_3+\mathbf e_{3}\otimes\mathbf e_\alpha)$,
hence
 \beeq\lab{Ju}\mathcal I(\u)=\int_\om \mathcal Q(\x, \mathbb E(\widetilde\u))\,d\x.\eneq
Moreover, by  Lemma \ref{load0} and Remark \ref{0L} we obtain
\beeq\lab{Lu}\mathcal L(\mathbf R \widetilde\u)=\mathcal L(\mathbf R\u)+\mathcal L(x_3(b^*_1\mathbf R\mathbf e_1+b^*_2\mathbf R\mathbf e_2))=\mathcal L(\mathbf R\u)
\qquad
\forall\, \mathbf R\in \mathcal S_{\mathcal L, E}.\eneq
Therefore by choosing
$$\widetilde{\mathbf R}\in \argmin\left\{ 
-\mathcal L(\mathbf R\widetilde\u): \mathbf R\in \mathcal S_{\mathcal L, E}\right\}$$
we get
\begin{equation}\label{uutilde}\widetilde{\mathcal G}^I(\u)=\int_\om \mathcal Q^I(\x,\mathbb E(\widetilde\u))\,d\x-\mathcal L(\widetilde{\mathbf R}\widetilde\u).\end{equation}

Since  $\partial\Omega$ has a finite number of connected components and $\dv \widetilde\u=0$ a.e. in $\om$ then by \cite[Proposition 3.1, Corollary 3.2]{KMPT} there exists a Sobolev  extension of $\widetilde\u$, still denoted with $\widetilde\u\in W^{1,p}(\R^3;\mathbb R^3)$, such that  $\spt\widetilde\u$ is compact
and $\dv\widetilde\u= 0$  in an open neighborhood $\om'$ of $\overline\om$. Since $\widetilde\u\in W^{1,p}(\R^3;\mathbb R^3)$ with $p > 3$ and $\spt\widetilde\u$ is compact
 then $\widetilde\u\in C^{0,\gamma}(\R^3;\R^3)\cap L^{\infty}(\R^3;\R^3)$ \MMM with $\gamma=1-3/p$  \KKK and we denote with 
 $$\|\widetilde\u\|_{0,\gamma}:=\sup_{\mathbb R^3}|\widetilde\u|+\sup_{{{\x\in\mathbb R^3,  \mathbf y\in\mathbb R^3}\atop{ \x\neq \mathbf y}}}\frac{|\widetilde \u(\x)-\widetilde \u(\mathbf y)|}{|\x-\mathbf y|^\gamma}$$
  its H\"older norm.
 Let now, for every $j\in\mathbb N$, $\widetilde\u_j:=\widetilde\u*\rho_j$, where $\rho_j(\x):=\eps_j^{-3}\rho(\eps_j^{-1}|\x|)$ and $\rho$ is the unit symmetric mollifier, so  that  $\spt\rho\subset B_{1}(0)$, and 	\MMM
 $$\eps_j:=h_j^{\gamma/2}.$$ 	\KKK
  It is well known that $\widetilde\u_j\in C^1_c(\R^3;\R^3)$, that $\widetilde\u_j\to \widetilde\u$  in $W^{1,p}(\R^3;\mathbb R^3)$, that $\displaystyle\|\widetilde\u_j\|_{L^\infty}\le \|\widetilde\u\|_{L^\infty}$ and that
 \beeq\lab{graduj}\displaystyle\|\nabla\widetilde\u_j\|_{L^\infty}\le \|\widetilde\u\|_{L^\infty}\!\!\int_{B_{\eps_j}(0)}\!\!|\nabla\rho_j(\yy)|\,d\yy=K\eps_j^{-1}\|\widetilde\u\|_{L^\infty},\eneq
 where we have set
 $$ K:= 4\pi\!\!\int_0^1|\rho'(r)|r^2\,dr.$$
  Moreover it is readily seen that $\dv\widetilde\u_j=0$  in an open neighbourhood $\om_*$ of $\overline\om$ such that $\overline{\om_*}\subset\om'$  and that  for every $\xx,\mathbf q\in \R^3$ we have
 \beeq\lab{estsup}\begin{aligned}
\displaystyle |\widetilde\u_j(\x)-\widetilde\u(\x)|&\le \int_{B_{\eps_j}(0)}|\widetilde\u(\x-\yy)-\widetilde\u(\x)|\rho_j(\yy)\,d\yy\\
&\displaystyle\le \|\widetilde\u\|_{0,\gamma}\int_{B_{\eps_j}(0)}|\yy|^{\gamma}|\rho_j(\yy)\,d\yy\le \eps_j^{\gamma}\|\widetilde\u\|_{0,\gamma},\\
\end{aligned}
\eneq

\beeq\lab{estgrad}\begin{aligned}
\displaystyle |\nabla\widetilde\u_j(\x)-\nabla\widetilde\u_j(\mathbf q)|&\le \int_{B_{\eps_j}(0)}|\widetilde\u(\x-\yy)-\widetilde\u(\mathbf q-\yy)||\nabla\rho_j(\yy)|\,d\yy\\
&\displaystyle\le \|\widetilde\u\|_{0,\gamma}|\xx-\mathbf q|^{\gamma}\int_{B_{\eps_j}(0)}|\nabla\rho_j(\yy)|\,d\yy=\MMM  K\eps_j^{-1}\KKK\|\widetilde\u\|_{0,\gamma}|\xx-\mathbf q|^{\gamma}.
\end{aligned}
\eneq

{\bf Step 2.}
Let now $\zz_j$ be the Lagrangian flow associated to the vector field $\widetilde\u_j$ as in Lemma \ref{newflowlemma}.
 Let
\beeq\lab{zj}   \ww_j(\x):=\zz_j(h_j,\x),\qquad \yy_j(\x):=\widetilde {\mathbf R}\ww_j(\x)\MMM+\beta_j\mathbf e_3\KKK,\eneq
where
\[
\beta_j=\beta_j(\|\widetilde \u\|_{0,\gamma}):=h_j\|\widetilde \u\|_{0,\gamma}\left(\eps_j^\gamma+\exp\left(Kh_j\eps_j^{-1}\|\widetilde\u\|_{0,\gamma}\right)-1\right).
\]
\MMM By \eqref{nuova1} and by \eqref{graduj}  
 we obtain
$$ \sup_{\x\in\overline\Omega}|\zz_j(h_j,\x)-\x|\le h_j\|\widetilde\u\|_{{\infty}}\exp\{K\|\widetilde\u\|_{{\infty}}h_j\eps_j^{-1}\}\quad\mbox{for every $j$,}$$
  where the right hand side goes to zero as $j\to+\infty$ since $h_j\eps_j^{-1}\to 0$, so that
 for every large enough $j$ we obtain
 $$\zz_j(h_j,\x)\in \om_*\quad\mbox{ for every $\x\in\overline \Omega$}$$ and, as in the proof of Lemma \ref{newflowlemma},  
$$\det\nabla \ww_j(\x)=\det\nabla \yy_j(\x)=1\quad\mbox{ for every $\x\in\overline \Omega$}.$$ \KKK 
\MMM
Moreover, still referring to the proof of Lemma \ref{newflowlemma}, and in particular to \eqref{nuovaextra}\KKK, we obtain \MMM
\begin{equation}\label{sommaesottrai}\begin{aligned}
&\left|h_j^{-1}(\nabla \zz_j(h_j,\x)-\mathbf I)-\nabla \widetilde\u_j(\x)\right|\le h_j^{-1}\int_0^{h_j}|\nabla\widetilde \u_j(\x)\nabla\zz_j(s,\x)-\nabla\widetilde\u_j(\x)|\,ds\\
&\qquad+h_j^{-1}\int_0^{h_j}|\nabla \widetilde\u_j(\zz_j(s,\x))\nabla\zz_j(s,\x)-\nabla\widetilde\u_j(\x)\nabla\zz_j(s,\x))|\,ds=:M_j(\x)+N_j(\x)
\end{aligned}\end{equation}
for every $\x\in\overline\om$ and we next estimate both terms:  from \eqref{flux3} and \eqref{graduj} we get
\[
\begin{aligned}
M_j(\x)&\le  h_j^{-1}\|\nabla\widetilde \u_j\|_\infty\int_0^{h_j}|\nabla\zz_j(s,\x)-\mathbf I|\,ds\le h_j^{-1}\|\nabla\widetilde \u_j\|_\infty\int_0^{h_j}3\left(\exp(s\|\nabla\widetilde \u_j\|_\infty)-1\right)\,ds\\&\le 3K\eps_j^{-1}\|\widetilde \u\|_\infty\left(\exp(Kh_j\eps_j^{-1}\|\widetilde\u\|_\infty)-1\right)
\end{aligned}
\]
while \eqref{estgrad}, 	\eqref{nuova1} and \eqref{nuova2} \MMM   entail
\[
\begin{aligned}
N_j(\x)&\le h_j^{-1}\int_0^{h_j}|\nabla\zz_j(s,\x)||\nabla\widetilde \u_j(\zz_j(s,\x))-\nabla\widetilde \u_j(\x)|\,ds\\
&\le 3K\eps_j^{-1}\|\widetilde\u\|_{0,\gamma} h_j^{-1}\int_0^{h_j}\exp(s\|\nabla\widetilde\u_j\|_\infty)\,|\zz_j(s,\x)-\x|^\gamma\,ds\\
&\le 3K\eps_j^{-1}\|\widetilde\u\|_{0,\gamma}\|\widetilde \u\|_\infty^\gamma h_j^{-1}\int_0^{h_j}s^\gamma\,\exp(s(1+\gamma)\|\nabla\widetilde\u_j\|_\infty)\,ds\\
&\le 3K\|\widetilde \u\|_{0,\gamma}^{1+\gamma}h_j^{\gamma}\eps_j^{-1} \,\exp\left(K(1+\gamma)h_j\eps_j^{-1}\|\widetilde\u\|_\infty\right).
\end{aligned}
\]
Since $\eps_j=h_j^{\gamma/2}$ and $\gamma\in(0,1)$, we se that both $\sup_{\x\in\overline\om}M_j(\x)$ and $\sup_{\x\in\overline\om}N_j(\x)$ vanish as $j\to+\infty$, so that \KKK
 by setting $\widetilde\vv_j(\x):= h_j^{-1}(\mathbf z_j(h_j,\x)-\x)$, from \eqref{sommaesottrai} we get
$\nabla\widetilde\vv_j-\nabla \widetilde\u_j\to 0$ in  $L^{\infty}(\om;\R^{3\times 3})$,
thus
$\nabla\widetilde\vv_j\to \nabla\widetilde\u$ in $ L^p(\om;\R^{3\times 3}),$
and by taking \eqref{flux2} into account we have $\widetilde \v_j\to\widetilde \u$ in $L^\infty(\Omega;\mathbb R^3)$ so that
$$\widetilde\vv_j\to \widetilde\u\qquad \hbox{in}\ W^{1,p}(\om;\R^{3}).$$
\MMM We also have $h_j\|\nabla\widetilde \v_j\|_\infty\to 0$, due to \eqref{graduj}, since $h_j\eps_j^{-1}\to 0$. \KKK

{\bf Step 3}. We check that $\mathbf y_j\in\mathcal A$. \MMM
By taking account that $\widetilde{\mathbf R}\in \mathcal S_{\mathcal L,E}\subset \{\mathbf R\in SO(3): \mathbf R\mathbf e_3=\mathbf e_3\}$, along with \eqref{flux2} and  \eqref{graduj}-\eqref{estsup}, we have from \eqref{zj}
\[\begin{aligned}
y_{j,3}-x_3&=w_{j,3}-x_3+\beta_j \ge h_j\widetilde u_{j,3} -h_j\|\widetilde \u\|_{{\infty}}(\exp(h_j\|\nabla\widetilde \u_j\|_{{\infty}})-1)+\beta_j\\
&=h_j(\widetilde u_{j,3}-\widetilde u_3)-h_j\|\widetilde \u\|_{{\infty}}(\exp(h_j\|\nabla\widetilde \u_j\|_{{\infty}})-1)+h_j\widetilde u_3+\beta_j\\
&\ge -h_j\eps_j^\gamma\|\widetilde \u\|_{0,\gamma}-h_j\|\widetilde \u\|_{0,\gamma}(\exp(Kh_j\eps_j^{-1}\|\widetilde \u\|_{0,\gamma})-1)+h_j\widetilde u_3+\beta_j=h_j\widetilde u_3
\end{aligned}\]
for every $j$,
where the last equality follows from the definition of $\beta_j$. Hence, by recalling that $\widetilde u_3^*(\x)\ge 0$ for q.e. $\x\in E\subset \{x_3=0\}$, 
 we deduce $y^*_{j,3}\ge 0$ for q.e. $\x\in E$, that is $\mathbf y_j\in\mathcal A$ for every $j$.
\KKK

 
 {\bf Step 4.}
 We conclude by noticing that Lemma \ref{lemma alternative} entails
$\mathcal L(\widetilde{\mathbf R}\x-\x)=0$,
whence
\[\begin{aligned}
h_j^{-1}\mathcal L(\yy_j-\x)&=h_j^{-1}\mathcal L(\widetilde{\mathbf R}\w_j-\x)+h_j^{-1}\beta_j\mathcal L(\mathbf e_3)\\
&=h_j^{-1}\mathcal L(\widetilde{\mathbf R}\zz_j(h_j,\x)-\x)+o(1)=\mathcal L(\widetilde{\mathbf R}\widetilde\v_j)+o(1)\quad\mbox{as $j\to+\infty,$}
\end{aligned}\]
where we have used the fact that $h_j^{-1}\beta_j=o(1)$ as $j\to+\infty$, which follows from the definitions of $\eps_j$ and $\beta_j$. 
So by exploiting \MMM that $h_j\|\nabla\widetilde \v_j\|_\infty\to 0$ along with  \eqref{regW}, \KKK by taking account that for \MMM large enough $j$ there holds $\det\nabla\mathbf y_j=\det(\mathbf I+h_j\nabla\widetilde\v_j)=1$ \KKK in $\om$ and that  $\dv\widetilde\uu=0$ a.e. in $\om$, we get
\[\begin{aligned}
\displaystyle\limsup_{j\to+\infty}\mathcal G_{j}^I(\mathbf y_j)&\le
 \limsup_{j\to+\infty}\int_{\Omega} {h_j}^{-2}\left(\mathcal W(\x,\mathbf I+h_j\nabla\widetilde\v_j)- \mathcal Q(\x,\mathbb E(h_j\widetilde\v_j))\right)\,d\x\\&\qquad+
\displaystyle \limsup_{j\to+\infty}\left(h_j^{-2}\int_{\Omega} \mathcal Q(\x,\mathbb E(h_j\widetilde\v_j))\,d\x-h_j^{-1}\mathcal L(\mathbf y_j-\x)\right)\\
&\displaystyle \le \limsup_{j\to+\infty}\int_{\Omega}\omega(h_j|\nabla\widetilde\v_j|)|\nabla\widetilde\v_j|^2\,d\x+\limsup_{j\to+\infty}\int_{\Omega} \mathcal Q(\x,\mathbb E(\widetilde\v_j))\,d\x-\mathcal L(\widetilde{\mathbf R}\widetilde\u)\\
&\displaystyle = \int_{\Omega} \mathcal Q^I(\x,\mathbb E(\widetilde\uu))\,d\x-\mathcal L(\widetilde{\mathbf R}\widetilde\u),\end{aligned}
\]
which proves the result \MMM in view of \eqref{uutilde}.\KKK
\end{proof}
 
We are now in a position to prove our main theorem.
\begin{proofad1}
 If $(\overline{\mathbf y}_j)_{j\in\mathbb N}\subset H^1(\om,\mathbb R^3)$ is a minimizing sequence for $\mathcal G_j^I$ then 
 we may assume that $\mathcal G_j^I(\overline{\mathbf y}_j)\le \mathcal G_j^I(\mathbf x)+1=1$ for every $j$. 
Moreover,
if
 $\mathbf R_j$ belong to $\mathcal A(\overline{\mathbf y}_j)$   and $\overline{\mathbf c}_j$ is defined \MMM by the right hand sides of \eqref{calfa}, \eqref{c3}, then Lemma \ref{lemma compactness} \KKK entails that the sequence 
\begin{equation*} \vspace{-0.1cm}
 {\overline\u}_j(\x):={h_j}^{-1}\mathbf R_j^T\left\{\big({\overline\yy}_j\,-\,\mathbf R_j\xx\,-{\overline{\mathbf c}}_j\,\big)_{\alpha}\mathbf e_{\alpha}+({\overline y}_{j,3}-x_3)\mathbf e_3\right\}
\end{equation*} 
is bounded in $H^1(\om;\mathbb R^3)$. Therefore up to subsequences ${\overline\u}_j\to \overline\u$ weakly in $H^1(\Omega;\R^3),$ so, by Lemma~\ref{Lemma 3.8}, we have $\overline\u\in \mathcal A$, $\dv \overline \u=0$ a.e. in $\om$ and
\begin{equation}\lab{liminfextra}
\displaystyle\liminf_{j\to +\infty}{\mathcal G}^I_j({\overline{\mathbf y}}_j)\ge \widetilde{\mathcal G}^I(\overline\u). 
\end{equation} 
On the other hand, by Lemma \ref{upbd},  for every $\u_*\in W^{1,p}(\Omega,\mathbb R^3)\cap\mathcal A$ with  $p > 3,$ there exists 
  a sequence  $\mathbf y_j\in C^1(\overline\Omega,\mathbb R^3)$ such that
\beeq\label{newlim} \limsup_{j\to +\infty} {\mathcal G}_j^I(\mathbf y_j) \le \widetilde{\mathcal G}^I(\u_*).\eneq
Since
\beeq\lab{infGj}{\mathcal G}^I_j(\overline\yy_j)+o(1)=\inf_{H^1(\om,\mathbb R^3)}\mathcal G^I_{j}\le {\mathcal G}^I_j(\mathbf y_j)\qquad \mbox{as $j\to+\infty,$}
\eneq
 by passing to the limit as $j\to +\infty$, we get 
 \beeq\lab{reglimsup}\widetilde{\mathcal G}^I(\overline\u)\le\widetilde {\mathcal G}^I(\u_*)
 \qquad \mbox{ for every $\u_*\in W^{1,p}(\Omega,\mathbb R^3)\cap\mathcal A$ with  $p > 3$}.
 \eneq

Now fix  $\u\!\in\! \mathcal A$ such that $\dv \u=0$ a.e. in $\om$ and denote again  by $\u$ a Sobolev extension of $\u$ to the whole $\mathbb R^3$. By  \cite[Lemma A.11]{MPTOB}  there exists 
 $\v_j\in C^1(\R^3,\mathbb R^3)\cap\mathcal A$ such that $\v_j\to \u$ in $H^1(\om;\mathbb R^3)$ and by Bogowskii's Theorem \cite{BOG}, see also \cite[Theorem A.2]{JS},   there exists $\ww_j$ belonging to $W^{1,r}_0(\om;\R^3)$ for every $r\ge1$ \KKK such that
 \beeq\lab{bog1}
 \displaystyle \dv\ww_j=-\dv\v_j+|\om|^{-1}\int_\om \dv\v_j\,d\x\quad\hbox{in}\ \om,\eneq and
\beeq\lab{bog2}
\displaystyle \|\ww_j\|_{H^1(\om;\R^3)}\le C_*\, \left\|-\dv\v_j+|\om|^{-1}\int_\om \dv\v_j\,d\x\right\|_{L^2(\om)}
 \eneq
 for some suitable constant $C_*=C_*(\om)$. 
Therefore by setting
\[\displaystyle \u_j:= \v_j-\left (|\om|^{-1}\int_\om \dv\v_j\,d\x\right)x_3\mathbf e_3+\ww_j\]
 and by taking account of \eqref{bog1} we get $\dv\u_j=0$ a.e. in $\om$. {Moreover since $E\subset \partial\om\cap \{x_3=0\}$ and $\u_j,\ \v_j\in C^{0}(\overline\om;\R^3)$ we have  $\u_j(\x)=\v_j(\x)$ for every $\x\in E$, hence $\u_j\in \mathcal A$.} Eventually by \eqref{bog2} we get $\ww_j\to 0$ in $H^1(\om;\R^3)$. By recalling that   $\v_j\to \u$ in $H^1(\om;\R^3)$ and that $\dv\v_j\to 0$ in $L^2(\om)$, we have
$$\|\u_j-\u\|_{H^1(\om;\R^3)}\le \left\|\v_j-\u-\left (|\om|^{-1}\int_\om \dv\v_j\,d\x\right)x_3\mathbf e_3\right\|_{H^1(\om;\R^3)}+\|\ww_j\|_{H^1(\om;\R^3)}\to 0.$$
 By \eqref{reglimsup} we have
$\widetilde{\mathcal G}^I(\overline\u)\le \widetilde{\mathcal G}^I(\u_j),$
whence by Remark \ref{cont} we have
$$\displaystyle\widetilde{\mathcal G}^I(\overline\u)\le \lim_{j\to +\infty}\widetilde{\mathcal G}^I(\u_j)=\widetilde{ \mathcal G}^I(\u).
$$
The arbitrariness of $\u\in\mathcal A$ shows that
 $\overline\u\in\argmin_{H^{1}(\om;\R^3)} \widetilde{\mathcal G}^I$.

\MMM
We claim that  ${\mathcal G}_j^I(\overline{\mathbf y}_j)\to \widetilde{\mathcal G}^I(\overline\u)$:
indeed, by the previous arguments, if $p>3$ we have $\min_{H^{1}(\om;\R^3)}\widetilde{\mathcal G}^I=\inf_{W^{1,p}(\om;\R^3)}\widetilde{\mathcal G}^I$.  
Since \eqref{liminfextra}, \eqref{newlim}, \eqref{infGj} and \eqref{reglimsup} imply
\[
\min_{H^1(\om;\R^3)}\widetilde{\mathcal G}^I=\widetilde{\mathcal G}^I(\overline \u)\le \liminf_{j\to+\infty}\widetilde{\mathcal G}_j^I(\overline\yy_j)\le\limsup_{j\to+\infty}\widetilde{\mathcal G}_j^I(\overline\yy_j)\le \widetilde {\mathcal G}^I(\u_*)
\]
for every $\u_*\in \mathcal A$ such that $\u_*\in W^{1,p}(\om;\R^3)$, the claim follows.

\KKK

%

We are only left to show  that
$\min_{H^1(\om;\R^3)}\widetilde{\mathcal G}^I =\min_{H^1(\om;\R^3)}\mathcal G^I.$
To this aim we show first that for every $\u\in \mathcal A$ there exists $\widetilde\u\in \mathcal A$ such that $\mathcal G^I(\widetilde \u)=\widetilde{\mathcal G}^I(\u)$. Indeed if $\widetilde\u$ is defined as in \eqref{utilde} then by \eqref{Ju} and \eqref{Lu} we get
\beeq\label{G=Gtilde}\displaystyle \widetilde{\mathcal G}^I(\u)=\mathcal I^I(\u)-\max_{\mathbf R\in \mathcal S_{\mathcal L, E}}\mathcal L(\mathbf R\u)=
 \int_\om \mathcal Q^I(\x,\mathbb E(\widetilde\u))\,d\x-\max_{\mathbf R\in \mathcal S_{\mathcal L, E}}\mathcal L(\mathbf R\widetilde\u)=\mathcal G^I(\widetilde\u)
\eneq
as claimed. By recalling that  $\widetilde{\mathcal G}^I(\overline\u)=\min_{H^1(\om;\R^3)} \widetilde{\mathcal G}^I\le \inf_{H^1(\om;\R^3)} \mathcal G^I$ let us assume that inequality is strict. Then
 by \eqref{G=Gtilde} there exists $\widetilde{{\overline\u}}\in \mathcal A$ such that $\mathcal G^I(\widetilde{{\overline\u}})=\widetilde{\mathcal G}^I(\overline\u) < \inf_{H^1(\om;\R^3)} \mathcal G^I$, a contradiction. Thus again by~\eqref{G=Gtilde} $\mathcal G^I(\widetilde{{\overline\u}})=\widetilde{\mathcal G}^I(\overline\u)= \min_{H^1(\om;\R^3)} \mathcal G^I$.
\end{proofad1}

\section*{Acknowledgements}  \noindent
 EM and DP are supported by the MIUR-PRIN project 202244A7YL. They  are members of the GNAMPA group of the Istituto Nazionale di Alta Matematica (INdAM).

 \end{document}